\documentclass{jams-l}

\usepackage{amssymb}

\usepackage{dsfont}

\usepackage{mathabx}

\usepackage{graphicx}

\usepackage{hyperref}

\usepackage{xcolor}

\usepackage{mathbbol}

\usepackage[autostyle]{csquotes}

\newtheorem{theorem}{Theorem}[section]
\newtheorem{corollary}[theorem]{Corollary}

\newtheorem{lemma}[theorem]{Lemma}

\theoremstyle{definition}

\newtheorem{example}[theorem]{Example}

\theoremstyle{remark}
\newtheorem{remark}[theorem]{Remark}

\numberwithin{equation}{section}

\DeclareSymbolFontAlphabet{\amsmathbb}{AMSb}%

\DeclareMathOperator*{\argmax}{arg\,max}
\DeclareMathOperator*{\argmin}{arg\,min}

\usepackage{tikz}

\usetikzlibrary{arrows,positioning,calc} 
\usetikzlibrary{decorations.pathreplacing}
\tikzset{
    >=stealth',
    punkt/.style={
           rectangle,
           rounded corners,
           draw=black, thick,
           text width=5.5em,
           minimum height=2em,
           text centered},
    punktl/.style={
           rectangle,
           rounded corners,
           draw=black, thick,
           text width=7em,
           minimum height=2em,
           text centered},
    pil/.style={
           ->,
           shorten <=4pt,
       shorten >=4pt
    },
    pildotted/.style={
           ->,
           shorten <=4pt,
           shorten >=4pt,
  dotted,
  },
    punktf/.style={
           rectangle,
           text width=4.0em,
           minimum height=1.5em,
           text centered},
    punktfTop/.style={
           rectangle,
           text width=4.0em,
           minimum height=1.5em,
           text centered,
           append after command={
               [thick,shorten >=0.2bp, shorten <=0.2bp]
               (\tikzlastnode.north west)edge(\tikzlastnode.north east)
}
    },
    punktfBot/.style={
           rectangle,
           text width=4.0em,
           minimum height=1.5em,
           text centered,
           append after command={
               [thick,shorten >=0.2bp, shorten <=0.2bp]
               (\tikzlastnode.south west)edge(\tikzlastnode.south east)
            }
    }
}

\begin{document}

\title[Expert Kaplan--Meier estimation]{Expert Kaplan--Meier estimation}


%    author one information
\author{Martin Bladt}
\address{Department of Mathematical Sciences, University of Copenhagen, Universitetsparken 5, DK-2100 Copenhagen, Denmark}
\curraddr{}
\email{\href{mailto:martinbladt@math.ku.dk}{martinbladt@math.ku.dk}}
\thanks{}

%    author two information
\author{Christian Furrer}
\address{Department of Mathematical Sciences, University of Copenhagen, Universitetsparken 5, DK-2100 Copenhagen, Denmark}
\curraddr{}
\email{\href{mailto:furrer@math.ku.dk}{furrer@math.ku.dk}}
\thanks{}

%    \subjclass is required.

\date{}

\dedicatory{\today}

%    Abstract is required.
\begin{abstract}
The setting of a right-censored random sample subject to contamination is considered. In various fields, expert information is often available and used to overcome the contamination. This paper integrates expert knowledge into the product-limit estimator in two different ways with distinct interpretations. Strong uniform consistency is proved for both cases under certain assumptions on the kind of contamination and the quality of expert information, which sheds light on the techniques and decisions that practitioners may take. The nuances of the techniques are discussed -- also with a view towards semi-parametric estimation -- and they are illustrated using simulated and real-world insurance data. \\[2mm]


\end{abstract}

\maketitle


\section{Introduction}

In many practical areas of application, observations might not only be right-censored but also subject to various forms of contamination. In health and disability insurance, for instance, some insurance claims are yet to be settled, while some claims are falsely settled and therefore likely to later resume. The former constitutes right-censoring, while the latter is a form of contamination. In this paper, we study how to cope with such contamination via the integration of expert beliefs to classic non-parametric survival analysis.

In statistics, a \textit{belief} is a commonly encountered concept with various contrasting definitions. Bayesian inference relies on prior information specified independently of the data, and although methods to automate the selection of priors exist, such as empirical Bayes methods, the Bayesian paradigm rests on model specifications being made based on beliefs. Frequentist inference also involves choices, such as what models and tests to use, which again projects beliefs into the statistical analysis of the data. In fact, every statistical model rests on beliefs, as argued in~\cite{Brownstein2019}, where the term \textit{expert knowledge} is used, and where it is further stressed that ignoring expert knowledge incurs considerable opportunity costs. 

The inclusion of expert knowledge, or \textit{expert information}, has been considered in a range of practice areas, ranging from insurance~\cite{Tredger2016} to climate~\cite{Mach2017}. Typically, beliefs come in the form of additional data. This data may either be of lower quality than the main part of the sample or it can possess some degree of subjectivity. Intuitively, we would like to incorporate expert information whenever it is fairly objective and disregard it { if too} subjective. However, in statistical discourse the term \textit{objectivity} is contentious, see~\cite{GelmanHennig2017}. Consequently, the line between data and expert information is rather blurred.

In actuarial mathematics, expert information emanates primarily from practicing actuaries, who take additional, sometimes non-statistical, factors into account to produce guesses on different aspects of their datasets. In case an insurance claim is fraudulent, the current datapoint likely exaggerates the amount actually to be paid to the insured; actuaries may be in possession of expert information about which claims are fraudulent and to what degree. For insurance claims that are yet to be settled (or closed), say in for instance third-party liability insurance, where claims typically take a long time to close, an actuary may also provide expert information on the eventual total claim size, confer with~\cite{Bladt2020}. How the expert information is obtained, say whether a statistical method was preliminarily used, and details concerning its exact level of accuracy is often missing.

In this paper, we study a problem of contamination. In health and disability insurance, an insured claims benefits for as long as they are disabled. Once the insured recovers, the claim is considered closed, and statistical analysis may be performed with respect to the duration of disability -- for instance to determine future insurance premiums. In the industry, it has been observed that some claims considered closed might reopen, that is to say, the original recovery might be annulled, so that the current datapoint understates the actual duration; such claims and datapoints are here said to be contaminated. We should like to stress that this form of contamination is not restricted to health and disability insurance. It also appears in general insurance, where the object of interest is the size and not the duration of the claim, and in randomized clinical trials under the term \textit{incomplete event adjudication}, see~\cite{CookKosorok2004}. 

Some types of temporary disability are rather long, so the datapoints are effectively right-censored at the time of observation. The result is that the dataset is subject to not only contamination but also right-censoring, with the latter only affecting fully observed claims. The starting point for our analysis is therefore classic, non-parametric survival analysis. This implies that the methods we explore may be applied to any survival analysis scenario for which the fully observed datapoints are at risk of being understated and an expert has knowledge about which datapoints have this feature or to what extent.

To be specific, we consider three random mechanisms underlying the data generation process: the mechanism of interest, a censoring mechanism, and a contamination mechanism. The first and second components are classic in the survival analysis literature, while the third component is to be overcome by the use of expert information. To this end, we adopt the \textit{inverse probability of censoring weighted} representation of the product-limit (Kaplan--Meier) estimator and incorporate judgments in the form of binary variables or kernel specifications. We provide conditions for strong consistency on compacts of the modified product-limit estimators, thus validating them as useful vehicles for a sound transfer of beliefs into the statistical framework.

The remainder of the paper is structured as follows. In Section~\ref{sec:km}, we review the definitions leading up to the Kaplan--Meier estimator, culminating in the retrieval of its inverse probability of censoring weighted representation. Section~\ref{sec:expert} presents the theoretical contributions of the paper: defining the expert Kaplan--Meier estimators, discovering suitable conditions on the quality of expert information for strong convergence on compacts, and comparing the resulting estimation procedures. The theory of empirical processes plays a central role and leads to a transparent treatment. Section~\ref{sec:ext} is devoted to extensions towards semi-parametric estimation, in particular concerning tail behavior, while we also discuss the introduction of covariates. Numerical studies for both simulated and real data are presented in Section~\ref{sec:num}, showcasing the practical potential of the presented methodology in real-life situations. {Section~\ref{sec:con} concludes.}

\section{Kaplan--Meier estimator}\label{sec:km}

In this section, we give a brief introduction to the celebrated Kaplan--Meier estimator for right-censored data while adopting some nomenclature from actuarial science, which constitutes one of our main areas of application. Furthermore, as a conceptual as well as technical prelude to upcoming arguments, we recall an alternative representation of the Kaplan--Meier estimator as an inverse probability of censoring weighted average, and we recall the usage of empirical processes and cumulative hazard functions in proving strong consistency on compacts.

Let $(\Omega,\mathcal{F},\amsmathbb{P})$ be a background probability space, and let $X$ be a random variable with values in $[0,\infty)$. We think of $X$ as either the size or the duration of an insurance claim. In practice, observation of $X$ is subject to right-censoring. In other words, we observe
\begin{align}\label{eq:W_delta_classic}
W = X \wedge C, \quad \delta = \mathds{1}_{\{W = X\}},
\end{align}
where $C$ is another random variable with values in $[0,\infty]$ describing right-censoring. If $\delta = 1$, then the claim is said to be \textit{closed}, while if $\delta = 0$, then the claim is said to be \textit{open}. Denote by $F$ the cumulative distribution function of $X$, by $G$ the cumulative distribution function of $C$, and by $H$ the cumulative distribution function of $W$. We assume that $F$ is non-degenerate, while $G$ might even be degenerate at $+\infty$, which corresponds to no right-censoring.

Let $(X_i,C_i)_{i=1}^n$ be iid replicates of $(X,C)$, and define $W_i$ and $\delta_i$ according to~\eqref{eq:W_delta_classic}. The product-limit estimator {$\amsmathbb{F}^{(n)}$} of $F$ is given by
\begin{align*}
1 - \amsmathbb{F}^{(n)}(t)
=
1 - \amsmathbb{L}\big(t;(W_i,\delta_i)_{i=1}^n\big) = \prod_{W_{n:i} \leq t} \bigg(1 - \frac{\delta_{n:i}}{n-i+1} \bigg), \quad t\geq 0,
\end{align*}
where $0 \leq W_{n:1} \leq \cdots \leq W_{n:n}$ and $\delta_{n:i}$ are the corresponding $\delta$'s. In case of ties, closed claims ($\delta_i = 1$) are taken to precede open claims ($\delta_i = 0$). The above non-parametric maximum likelihood estimator was introduced by Kaplan and Meier in their seminal paper~\cite{KaplanMeier1958}, and is, accordingly, also called the Kaplan--Meier estimator. In the following, we rely on the comprehensive yet succinct presentation found in Chapter~7 of~\cite{ShorackWellner2009}.

Under the assumption of \textit{entirely random} right-censoring, that is, the assumption that $X$ and $C$ are independent, the Kaplan--Meier estimator is consistent in the following sense:
\begin{align}\label{eq:km_classic_cons}
\forall \theta < H^{-1}(1): \sup_{0 \leq t \leq \theta} \big| \amsmathbb{F}^{(n)}(t) - F(t) \big| \overset{\text{a.s.}}{\to} 0, \quad n \to \infty,
\end{align}
see for instance the proof of Theorem 7.3.1 in~\cite{ShorackWellner2009}. Note that the last part of the proof, which extends~\eqref{eq:km_classic_cons} from $0\leq t \leq \theta$ for any $\theta < H^{-1}(1)$ to $0\leq t < H^{-1}(1)$ is erroneous. This appears to have been first noted in~\cite{Wang1987} and later resolved
in~\cite{StuteWang1993}; see also the Errata of~\cite{ShorackWellner2009} and Gill's account starting from p.\ 151 in~\cite{BakryGillMolchanov1994}.

To prove~\eqref{eq:km_classic_cons}, it is convenient to cast the Kaplan--Meier estimator according to
\begin{align}\label{eq:F_emp}
\amsmathbb{F}^{(n)}(t)
&=
\int_{[0,t]} \big(1 - \amsmathbb{F}^{(n)}(s-)\big) \, \mathbb{\Lambda}^{(n)}(\mathrm{d}s), \\ \nonumber
\mathbb{\Lambda}^{(n)}(t)
&=
\int_{[0,t]} \frac{1}{1 - \amsmathbb{F}^{(n)}(s-)} \, \amsmathbb{F}^{(n)}(\mathrm{d}s), \quad t\geq0,
\end{align}
which are empirical equivalents of
\begin{align*}
F(t) &= \int_{[0,t]} \big(1 - F(s-)\big) \, \Lambda(\mathrm{d}s), \\
\Lambda(t) &= \int_{[0,t]} \frac{1}{1 - F(s-)} \, F(\mathrm{d}s), \quad t\geq0,
\end{align*}
where $\Lambda$ is the cumulative hazard function corresponding to $F$. Under entirely random right-censoring, we have that $(1 - F)(1 - G)=(1-H)$, so
\begin{align}\label{eq:lambda_entirely_random}
\Lambda(t) = \int_{[0,t]} \frac{1}{1 - H(s-)} \, H_1(\mathrm{d}s), \quad 0 \leq t < G^{-1}(1),
\end{align}
for $H_1(t) = \amsmathbb{P}(W \leq t, W = X)$. Referring to Exercise 7.2.2 in~\cite{ShorackWellner2009}, it also holds pathwise that
\begin{align}\label{eq:lambda_emp_H}
\mathbb{\Lambda}^{(n)}(t) = \int_{[0,t]} \frac{1}{1 - \amsmathbb{H}^{(n)}(s-)} \, \amsmathbb{H}^{(n)}_1(\mathrm{d}s),
\end{align}
where the empirical estimators $\amsmathbb{H}^{(n)}$ and $\amsmathbb{H}^{(n)}_1$ of $H$ and $H_1$ are given by
\begin{align*}
\amsmathbb{H}^{(n)}(t) &= \frac{1}{n} \sum_{i=1}^n \mathds{1}_{\{W_i \leq t\}}, \\
\amsmathbb{H}^{(n)}_1(t) &=\frac{1}{n} \sum_{i=1}^n \mathds{1}_{\{W_i \leq t\}}\delta_i
\end{align*}
for $t\geq0$. Consistency of the Kaplan--Meier estimator under entirely random right-censoring now essentially follows from the rather obvious consistency of the empirical estimators $\amsmathbb{H}^{(n)}$ and $\amsmathbb{H}^{(n)}_1$ and the more involved fact that, pathwise,
\begin{align}\label{eq:gill_rep}
\frac{1-\amsmathbb{F}^{(n)}(t)}{1-F(t)} = 1 - \int_{[0,t]} \frac{1 - \amsmathbb{F}^{(n)}(s-)}{1 - F(s)} \, \big(\mathbb{\Lambda}^{(n)} - \Lambda\big)(\mathrm{d}s), \quad 0 \leq t < F^{-1}(1),
\end{align}
see Proposition 7.2.1 in~\cite{ShorackWellner2009}. Note that entirely random right-censoring is only required to link $\Lambda$ to $H$ and $H_1$, confer with~\eqref{eq:lambda_entirely_random}.

The Kaplan--Meier estimator {$\amsmathbb{G}^{(n)}$} of $G$ is given by
\begin{align*}
1 - \amsmathbb{G}^{(n)}(t)
=
1 - \amsmathbb{L}'\big(t;(W_i,1-\delta_i)_{i=1}^n\big) = \prod_{W_{n:i} \leq t} \bigg(1 - \frac{1 - \delta_{n:i}}{n-i+1} \bigg), \quad t\geq 0,
\end{align*}
where in case of ties, closed claims $(1 - \delta_i = 0)$ are taken to precede open claims $(1 - \delta_i =1)$. Note that in case of no ties, it holds that
\begin{align*}
\amsmathbb{L}'\big(t;(W_i,1-\delta_i)_{i=1}^n\big) = \amsmathbb{L}\big(t;(W_i,1-\delta_i)_{i=1}^n\big).
\end{align*}
In any case, we should like to stress that under entirely random right-censoring, this estimator is also consistent:
\begin{align}\label{eq:km_censored_cons}
\forall \theta < H^{-1}(1): \sup_{0 \leq t \leq \theta} \big| \amsmathbb{G}^{(n)}(t) - G(t) \big| \overset{\text{a.s.}}{\to} 0, \quad n \to \infty.
\end{align}
Similar to~\eqref{eq:lambda_entirely_random}, under entirely random right-censoring it holds that
\begin{align}\label{eq:ipcw_rep}
F(t) = \int_{[0,t]} \frac{1}{1-G(s-)} \, H_1(\mathrm{d}s), \quad 0 \leq t < G^{-1}(1).
\end{align}
There actually exists a comparable representation of the Kaplan--Meier estimator $\amsmathbb{F}^{(n)}$ in terms of the Kaplan--Meier estimator of $G$ as well as $\amsmathbb{H}^{(n)}_1$. The argument goes as follows. Referring to Exercise~7.2.2 in~\cite{ShorackWellner2009}, one has the following pathwise identity:
\begin{align*}
\big(1-\amsmathbb{F}^{(n)}\big)\big(1-\amsmathbb{G}^{(n)}\big) = 1 - \amsmathbb{H}^{(n)},
\end{align*}
so that~\eqref{eq:lambda_emp_H} in combination with~\eqref{eq:F_emp} yields
\begin{align}\label{eq:emp_ipcw_rep}
\amsmathbb{F}^{(n)}(t) = \int_{[0,t]} \frac{1}{1-\amsmathbb{G}^{(n)}(s-)} \, \amsmathbb{H}_1^{(n)}(\mathrm{d}s) = \frac{1}{n} \sum_{i=1}^n \frac{\mathds{1}_{\{W_i \leq t\}}\delta_i}{1-\amsmathbb{G}^{(n)}(W_i-)},
\end{align}
which casts the Kaplan--Meier estimator as a so-called \textit{inverse probability of censoring weighted} average. {See Section~3.3 in~\cite{LaanRobins2003} for further details on inverse probability of censoring weighting.}

\section{Expert Kaplan--Meier estimators}\label{sec:expert}

For a range of applications, right-censored data is subject to a form of contamination. Let us again think of $X$ as either the size or the duration of an insurance claim and of $\delta$ as indicating whether an insurance claim is closed or open. In for instance motor third-party liability insurance and both worker's compensation and loss of income coverage, claims might potentially reopen, see for instance Subsections~K2-K3 in~\cite{Manual1997}. If this form of contamination is not taken into account, insurers might significantly misvalue their obligations. Broadly speaking, the potential of claims to reopen constitutes a form of incomplete event adjudication, which has received some attention in the context of randomized clinical trials, see~\cite{CookKosorok2004}.

We now extend the setup from Section~\ref{sec:km} to allow for contamination due to incomplete event adjudication. Let $Y$ be another random variable with values in $[0,\infty]$, and suppose that we observe only
\begin{align}\label{eq:W_delta_contaminated}
W = (X \wedge Y) \wedge C = X \wedge (Y \wedge C), \quad \delta = \mathds{1}_{\{W = X \wedge Y\}}.
\end{align}
In other words, $X$ is observed subject to both contamination and right-censoring. It should be noted that $Y$ is often somewhat hypothetical, as is also the case for ordinary competing risks models. We say that contaminated right-censoring is \textit{entirely random} if $X$ is independent of $Y \wedge C$. We say that both contamination and right-censoring are \textit{entirely random} if $X$, $Y$, and $C$ are mutually independent. Furthermore, if $\delta = 1$ and $X \leq Y$, then the claim is said to be \textit{truly closed}, while if $\delta = 1$ but $X > Y$, then the claim is said to be \textit{falsely closed}. The claim is still said to be \textit{open} if $\delta = 0$. See Figure~\ref{fig:cont_scheme} for a schematic representation of the contamination and censoring setup. The nomenclature hints at the main motivation for the present setup, namely insurance claims that might potentially reopen.

\begin{figure}[!htbp]
\centering
\includegraphics[width=0.7\textwidth,trim= 2in 2in 6in 2in,clip]{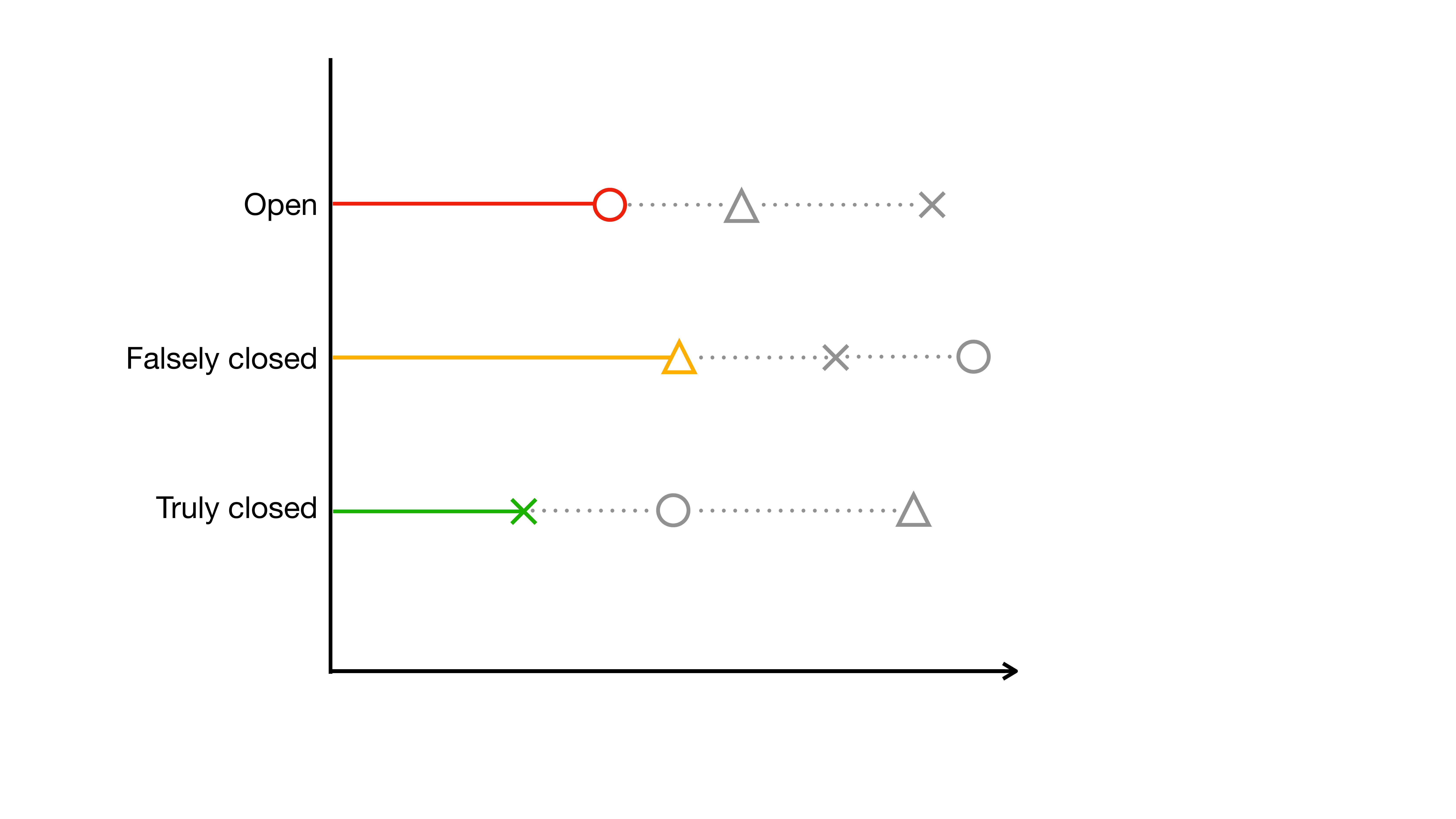}
\caption{Schematic representation of right-censored data subject to contamination. {The right-censoring mechanism is represented by a circle, the contamination mechanism by a triangle, and the true event of interest by a cross.}}
\label{fig:cont_scheme}
\end{figure}

Let $(X_i,Y_i,C_i)_{i=1}^n$ be iid replicates of $(X,Y,C)$, and define $W_i$ and $\delta_i$ according to~\eqref{eq:W_delta_contaminated}. The usual Kaplan--Meier estimator $t \mapsto \amsmathbb{L}(t; (W_i,\delta_i)_{i=1}^n)$ disregards the effect of contamination and thus -- unsurprisingly -- systematically overestimates the actual cumulative distribution function $F$ of $X$. But of course one can hardly do better without additional knowledge about the nature and extent of contamination. The alternative $\amsmathbb{L}(t; (W_i, \mathds{1}_{\{W_i = X_i\}})_{i=1}^n)$ is not viable since the indicator $\mathds{1}_{\{W_i = X_i\}}$ is not observed.

In the following, we develop non-parametric estimators that include expert information about exactly the nature and extent of contamination, which -- for insurance applications -- may be provided by caseworkers or claim analysts. {We do not make any assumptions on which data or even models the experts base their judgments on. This contrasts any related approach that would handle contamination by an explicit model for the reopening of claims, including so-called cure models; a recent survey on cure models is~\cite{AmicoKeilegom2018}.

We proceed as follows.} Subsection~\ref{sec:crude} is devoted to the crude expert, who provides judgments at $\mathds{1}_{\{W_i = X_i\}}$, while Subsection~\ref{sec:sophisticated} concerns the sophisticated expert, who provides beliefs regarding $t \mapsto \mathds{1}_{\{X_i \leq t\}}$. {Intuitively speaking, the crude expert selects truly closed claims among all truly and falsely closed claims, while the sophisticated expert estimates the true claim size or duration for all truly and falsely closed claims. In particular, the decision making process of the sophisticated expert must be more complex but, on the other hand, also significantly more informative.} In both cases, we establish strong consistency on compacts subject to suitable conditions on the quality of expert information. Finally, Subsection~\ref{sec:compare} contains a comparison of the two experts and the associated non-parametric estimators.

\subsection{Crude expert}\label{sec:crude}

The crude expert, which we refer to using the subscript $\dagger$, provides expert judgments $(\eta_i^{(n)})_{i=1}^n$, where $\eta_i^{(n)} \in \{0,1\}$ is their best judgment at $\mathds{1}_{\{W_i = X_i\}}$. The expert judgments may be viewed as a triangular array, and in that spirit, we assume that for every $n$ the variables $\eta_1^{(n)}, \ldots, \eta_n^{(n)}$ are mutually independent but not necessarily identically distributed.

Define the crude expert Kaplan--Meier estimator $\amsmathbb{F}_{\dagger}^{(n)}$ according to
\begin{align*}
\amsmathbb{F}_{\dagger}^{(n)}(t) = \amsmathbb{L}\big(t;(W_i,\eta_i^{(n)})_{i=1}^n\big), \quad t \geq0,
\end{align*}
or, equivalently,
\begin{align*}
\amsmathbb{F}_{\dagger}^{(n)}(t) = \frac{1}{n}\sum_{i=1}^n \frac{\mathds{1}_{\{W_i \leq t\}} \eta_i^{(n)}}{1 - \amsmathbb{L}'\big(W_i-;(W_i,1-\eta_i^{(n)})_{i=1}^n\big)}.
\end{align*}
This estimator is based on the ideal yet unavailable Kaplan--Meier estimator $t \mapsto \amsmathbb{L}(t; (W_i, \mathds{1}_{\{W_i = X_i\}})_{i=1}^n)$, but with the unobservables $(\mathds{1}_{\{W_i = X_i\}})_{i=1}^n$ replaced by the expert judgments $(\eta_i^{(n)})_{i=1}^n$. {Any reasonable expert would set $\eta_i^{(n)} \leq \delta_i$, in which case the usual Kaplan--Meier estimator, which corresponds to setting $\eta_i^{(n)} = \delta_i$, constitutes an upper bound for $\amsmathbb{F}_{\dagger}^{(n)}$.}

From~\eqref{eq:gill_rep}, it holds pathwise that
\begin{align}\label{eq:F_dagger_representation}
\frac{1-\amsmathbb{F}_{\dagger}^{(n)}(t)}{1-F(t)}
&=
1-\int_{[0,t]} \frac{1-\amsmathbb{F}_{\dagger}^{(n)}(s-)}{1-F(s)}\big(\mathbb{\Lambda}_{\dagger}^{(n)} - \Lambda\big)(\mathrm{d}s), \quad 0 \leq t < F^{-1}(1), \\ \label{eq:Lambda_dagger_representation}
\mathbb{\Lambda}_{\dagger}^{(n)}(t)
&=
\int_{[0,t]} \frac{1}{1-\amsmathbb{H}^{(n)}(s-)} \amsmathbb{H}_{\dagger,1}^{(n)}(\mathrm{d}s), \quad t \geq 0,
\end{align}
where $\amsmathbb{H}_{\dagger,1}^{(n)}$ is given by
\begin{align*}
\amsmathbb{H}_{\dagger,1}^{(n)}(t) = \frac{1}{n} \sum_{i=1}^n \mathds{1}_{\{W_i \leq t\}} \eta_i^{(n)}, \quad t\geq0.
\end{align*}
\begin{lemma}\label{lem:crude_expert}
Suppose that
\begin{align}\label{eq:crude_assumption}
\sup_{0 \leq t < \infty} \bigg|
\frac{1}{n}
\sum_{i=1}^n \amsmathbb{E}\big[\mathds{1}_{\{W_i \leq t\}} \eta_i^{(n)}\big] - H_1(t)\bigg| \to 0, \quad n \to \infty.
\end{align}
It then holds that
\begin{align*}
\sup_{0 \leq t < \infty}\big| \amsmathbb{H}_{\dagger,1}^{(n)}(t) - H_1(t) \big| \overset{\text{a.s.}}{\to} 0, \quad n \to \infty.
\end{align*}
\end{lemma}
\begin{proof}
We immediately find that
\begin{align*}
\sup_{0 \leq t < \infty}\big| \amsmathbb{H}_{\dagger,1}^{(n)}(t)  - H_1(t) \big|
&\leq
\sup_{0 \leq t < \infty} \bigg| \frac{1}{n} \sum_{i=1}^n \mathds{1}_{\{W_i \leq t\}} \eta_i^{(n)} - \frac{1}{n} \sum_{i=1}^n \amsmathbb{E}\big[\mathds{1}_{\{W_i \leq t\}} \eta_i^{(n)}\big] \bigg| \\
&\quad +
\sup_{0 \leq t < \infty} \bigg| \frac{1}{n} \sum_{i=1}^n \amsmathbb{E}\big[\mathds{1}_{\{W_i \leq t\}} \eta_i^{(n)}\big] - H_1(t) \bigg|.
\end{align*}
The last term converges to zero as $n \to \infty$ due to~\eqref{eq:crude_assumption}. Concerning the other term, we note that
\begin{align*}
&\sup_{0 \leq t < \infty} \bigg| \frac{1}{n} \sum_{i=1}^n \mathds{1}_{\{W_i \leq t\}} \eta_i^{(n)} - \frac{1}{n} \sum_{i=1}^n \amsmathbb{E}\big[\mathds{1}_{\{W_i \leq t\}} \eta_i^{(n)}\big] \bigg| \\
&\leq
\sup_{0 \leq t < \infty} \bigg| \frac{1}{n} \sum_{i=1}^n \mathds{1}_{\{W_i\eta_i^{(n)} \leq t\}} - \frac{1}{n} \sum_{i=1}^n \amsmathbb{P}\big(W_i \eta_i^{(n)} \leq t\big) \bigg|
+
\bigg| \frac{1}{n} \sum_{i=1}^n \eta_i^{(n)} - \amsmathbb{E}\big[\eta_i^{(n)}\big] \bigg|.
\end{align*}
The first term on the second line converges almost surely to zero as $n\to\infty$ according to the Glivenko--Cantelli theorem for triangular arrays, confer for instance with Theorem~3.2.1 in~\cite{ShorackWellner2009}. Furthermore, since $|\eta_i^{(n)}| \leq 1$,  the second term also converges almost surely to zero according to the strong law of large numbers for triangular arrays. Collecting results completes the proof.
\end{proof}
\begin{remark}
If the expert judgments are identically distributed, then~\eqref{eq:crude_assumption} reads
\begin{align*}
\sup_{0\leq t < \infty} \Big| \amsmathbb{E}\big[\mathds{1}_{\{W_1 \leq t\}} \eta_1^{(n)}\big] - H_1(t)\Big| \to 0, \quad n \to \infty.
\end{align*}
\end{remark}
\begin{remark}
A sufficient but not necessary condition for~\eqref{eq:crude_assumption} to hold is
\begin{align*}
\sup_{1 \leq i \leq n} 
\amsmathbb{E}\big[\big|\eta_i^{(n)} - \mathds{1}_{\{W_i = X_i\}}\big|\big] \to 0, \quad n \to \infty.
\end{align*}
To see this, note that
\begin{align*}
&\sup_{0 \leq t < \infty} \bigg| \frac{1}{n}\sum_{i=1}^n \amsmathbb{E}\big[\mathds{1}_{\{W_i \leq t\}} \eta_i^{(n)}\big] - H_1(t)\bigg| \\
&=
\sup_{0 \leq t < \infty} \bigg| \frac{1}{n} \sum_{i=1}^n \amsmathbb{E}\big[\mathds{1}_{\{W_i \leq t\}} \big(\eta_i^{(n)} - \mathds{1}_{\{W_i = X_i\}}\big)\big]\bigg| \\
&\leq
\frac{1}{n} \sum_{i=1}^n \amsmathbb{E}\big[\big|\eta_i^{(n)} - \mathds{1}_{\{W_i = X_i\}}\big|\big] \\
&\leq
\sup_{1 \leq i \leq n} \amsmathbb{E}\big[\big|\eta_i^{(n)} - \mathds{1}_{\{W_i = X_i\}}\big|\big].
\end{align*}
\end{remark}
Suppose for an instance that both contamination and right-censoring are entirely random, that is, suppose that $X$, $Y$, and $C$ are mutually independent. Inspired by the condition~\eqref{eq:crude_assumption}, it could be that the crude expert provides judgments based on their learnings about the following conditional expectation:
\begin{align}\label{eq:p_function}
\amsmathbb{E}\big[\mathds{1}_{\{W = X\}} \, \big| \, \delta, W \big]
=
\delta\amsmathbb{E}\big[\mathds{1}_{\{X \leq Y\}} \, \big| \, X \wedge Y] = \delta p(X \wedge Y).
\end{align}
This is somewhat in the spirit of~\cite{CookKosorok2004}, where incomplete event adjudication is studied. {Th}ere it is essentially suggested to first learn $p$ to form an estimator $\mathbb{p}^{(n)}$, and, then, to learn $F$ via the judgments
\begin{align*}
\eta_i^{(n)} =
\delta_i \, \mathbb{p}^{(n)}(X_i \wedge Y_i) \in [0,1].
\end{align*}
{Different from~\cite{CookKosorok2004}, our setup uses a competing risks representation of the contamination mechanism with rather explicit independence assumptions. Differences in assumptions are merely the result of using different but essentially equivalent representations of the data generating process.}

The next result establishes strong consistency on compacts for the crude expert Kaplan--Meier estimator subject to~\eqref{eq:crude_assumption}.
\begin{theorem}[Strong uniform consistency of the crude expert]\label{thm:cons_crude}
Suppose contaminated right-censoring is entirely random. Let $0\leq\theta<H^{-1}(1)$. Under~\eqref{eq:crude_assumption} it then holds that
\begin{align*}
&\sup_{0 \leq t \leq \theta}\big| \mathbb{\Lambda}_{\dagger}^{(n)}(t) - \Lambda(t) \big| \overset{\text{a.s.}}{\to} 0, \quad n \to \infty, \\
&\sup_{0 \leq t \leq \theta}\big| \amsmathbb{F}_\dagger^{(n)}(t) - F(t) \big| \overset{\text{a.s.}}{\to} 0, \quad n \to \infty.
\end{align*}
\end{theorem} 
\begin{proof}
This proof is in the spirit of the proof of Theorem~7.3.1 in~\cite{ShorackWellner2009}. Since $X$ is independent of $Y \wedge C$, the representations~\eqref{eq:lambda_entirely_random} and~\eqref{eq:Lambda_dagger_representation} yield
\begin{align}\label{eq:lambda_dagger_two_cases}
\begin{split}
\sup_{0 \leq t \leq \theta}\big| \mathbb{\Lambda}_{\dagger}^{(n)}(t) - \Lambda(t) \big|
&\leq
\sup_{0 \leq t \leq \theta}\bigg| \int_{[0,t]} \Big(\frac{1}{1-\amsmathbb{H}^{(n)}(s-)}-\frac{1}{1-H(s-)}\Big) \amsmathbb{H}_{\dagger,1}^{(n)}(\mathrm{d}s) \bigg| \\
&\quad +
\sup_{0 \leq t \leq \theta}\bigg| \int_{[0,t]} \frac{1}{1 - H(s-)} \big(\amsmathbb{H}_{\dagger,1}^{(n)} - H_1\big)(\mathrm{d}s) \bigg|.
\end{split}
\end{align}
We first take a closer look at the last term. Using integration by parts one may show that for $0 \leq t \leq \theta$,
\begin{align*}
&\bigg|\int_{[0,t]} \frac{1}{1 - H(s-)} \big(\amsmathbb{H}_{\dagger,1}^{(n)} - H_1\big)(\mathrm{d}s) \bigg| \\
&\leq
\bigg|\frac{1}{1-H(t-)} \big(\amsmathbb{H}_{\dagger,1}^{(n)}(t) - H_1(t)\big)\bigg|
+
\bigg|\int_{[0,t]} \big( \amsmathbb{H}_{\dagger,1}^{(n)}(s) - H_1(s) \big) \mathrm{d}\Big(\frac{1}{1 - H(s)}\Big) \bigg|\\
&\leq
\frac{1}{1-H(\theta)} \big|\amsmathbb{H}_{\dagger,1}^{(n)}(t) - H_1(t)\big|
+
\frac{1}{1-H(\theta)}\sup_{0 \leq s \leq t} \big| \amsmathbb{H}_{\dagger,1}^{(n)}(s) - H_1(s) \big|.
\end{align*}
Thus the last term of~\eqref{eq:lambda_dagger_two_cases} converges almost surely to zero as $n\to\infty$ according to Lemma~\ref{lem:crude_expert}. Concerning the other term, we note that
\begin{align*}
&\sup_{0 \leq t \leq \theta}\bigg| \int_{[0,t]} \Big(\frac{1}{1-\amsmathbb{H}^{(n)}(s-)}-\frac{1}{1-H(s-)}\Big) \amsmathbb{H}_{\dagger,1}^{(n)}(\mathrm{d}s) \bigg| \\
&\leq
\sup_{0 \leq t \leq \theta}\Big| \frac{1}{1-\amsmathbb{H}^{(n)}(s-)}-\frac{1}{1-H(s-)}\Big|.
\end{align*}
The ordinary Glivenko--Cantelli theorem yields
\begin{align*}
\sup_{0 \leq t < \infty}\big| \amsmathbb{H}^{(n)}(t) - H(t) \big| \overset{\text{a.s.}}{\to} 0, \quad n \to \infty,
\end{align*}
so that also
\begin{align*}
\sup_{0 \leq t < \infty}\big| \amsmathbb{H}^{(n)}(t-) - H(t-) \big| \overset{\text{a.s.}}{\to} 0, \quad n \to \infty,
\end{align*}
where we rely on the fact that $\amsmathbb{H}^{(n)}$ and $H$ are càdlàg. Since $[0,G(\theta-)] \ni x \mapsto \frac{1}{1-x}$ is uniformly continuous, this implies
\begin{align*}
\sup_{0 \leq t \leq \theta}\Big| \frac{1}{1-\amsmathbb{H}^{(n)}(s-)}-\frac{1}{1-H(s-)}\Big| \overset{\text{a.s.}}{\to} 0, \quad n \to \infty,
\end{align*}
Collecting results establishes the first assertion of the theorem. Furthermore, since $\mathbb{\Lambda}_{\dagger}^{(n)}$ and $\Lambda$ are càdlàg, it also holds that
\begin{align}\label{eq:lambda_crude_left_conv}
\sup_{0 \leq t \leq \theta}\big| \mathbb{\Lambda}_{\dagger}^{(n)}(t-) - \Lambda(t-) \big| \overset{\text{a.s.}}{\to} 0, \quad n \to \infty.
\end{align}
We now turn our attention to the second assertion of the theorem. It follows from~\eqref{eq:F_dagger_representation} and integration by parts that for $0\leq t \leq \theta$,
\begin{align*}
\big| \amsmathbb{F}_\dagger^{(n)}(t) - F(t) \big|
&=
\big(1-F(t)\big) \bigg| \int_{[0,t]} \frac{1-\amsmathbb{F}_{\dagger}^{(n)}(s-)}{1-F(s)}\big(\mathbb{\Lambda}_{\dagger}^{(n)} - \Lambda\big)(\mathrm{d}s) \bigg| \\
&\leq
\bigg| \int_{[0,t]} \frac{1-\amsmathbb{F}_{\dagger}^{(n)}(s-)}{1-F(s-)}\frac{1}{1-\Delta\Lambda(s)}\big(\mathbb{\Lambda}_{\dagger}^{(n)} - \Lambda\big)(\mathrm{d}s) \bigg| \\
&\leq 
 \frac{1-\amsmathbb{F}_{\dagger}^{(n)}(t)}{1-F(t)} \big|\amsmathbb{K}^{(n)}(t)\big|
+
\bigg| \int_{[0,t]} \amsmathbb{K}^{(n)}(s) \Big(\frac{1-\amsmathbb{F}_{\dagger}^{(n)}}{1-F}\Big)(\mathrm{d}s) \bigg| \\
&\leq
 \frac{1}{1-F(\theta)} \big|\amsmathbb{K}^{(n)}(t)\big|
+
\bigg| \int_{[0,t]} \amsmathbb{K}^{(n)}(s) \Big(\frac{1-\amsmathbb{F}_{\dagger}^{(n)}}{1-F}\Big)(\mathrm{d}s) \bigg|,
\end{align*}
where the auxiliary quantity $\amsmathbb{K}^{(n)}$ is given by
\begin{align*}
\amsmathbb{K}^{(n)}(\tau) = \int_{[0,\tau]} \frac{1}{1-\Delta\Lambda(s)} \big(\mathbb{\Lambda}_{\dagger}^{(n)} - \Lambda\big)(\mathrm{d}s), \quad 0 \leq \tau \leq \theta.
\end{align*}
Integration by parts also yields
\begin{align*}
\Big(\frac{1-\amsmathbb{F}_{\dagger}^{(n)}}{1-F}\Big)(\mathrm{d}s)
=
\frac{1}{1-F(s-)}\amsmathbb{F}_{\dagger}^{(n)}(\mathrm{d}s) - \frac{\amsmathbb{F}_{\dagger}^{(n)}(s)}{\big(1-F(s-)\big)\big(1-F(s)\big)} F(\mathrm{d}s)
\end{align*}
on $[0,\theta]$. In particular, the total variation of
\begin{align*}
[0,\theta] \ni \tau \mapsto 
\frac{1-\amsmathbb{F}_{\dagger}^{(n)}(\tau)}{1-F(\tau)}
\end{align*}
is bounded by $2{(1-F(\theta))}^{-2}$. Collecting results yields
\begin{align*}
\big| \amsmathbb{F}_\dagger^{(n)}(t) - F(t) \big|
\leq
\frac{3}{\big(1-F(\theta)\big)^2} \sup_{0 \leq s \leq t} \big|\amsmathbb{K}^{(n)}(s) \big|,
\end{align*}
so that also
\begin{align*}
\sup_{0 \leq t \leq \theta} \big| \amsmathbb{F}_\dagger^{(n)}(t) - F(t) \big|
\leq
\frac{3}{\big(1-F(\theta)\big)^2} \sup_{0 \leq t \leq \theta} \big|\amsmathbb{K}^{(n)}(t) \big|.
\end{align*}
It remains to be shown that $\sup_{0 \leq t \leq \theta} |\amsmathbb{K}^{(n)}(t)|$ converges almost surely to zero as $n\to \infty$. Note that
\begin{align*}
\sup_{0 \leq t \leq \theta} \big|\amsmathbb{K}^{(n)}(t) \big|
&=
\sup_{0 \leq t \leq \theta} \bigg| \int_{[0,t]} \frac{\Delta F(s)}{1-F(s)} \big(\mathbb{\Lambda}_{\dagger}^{(n)} - \Lambda\big)(\mathrm{d}s) + \int_{[0,t]} \big(\mathbb{\Lambda}_{\dagger}^{(n)} - \Lambda\big)(\mathrm{d}s) \bigg| \\
&\leq
\sum_{0 \leq t \leq \theta \atop \Delta F(t) > 0} \frac{\Delta F(t)}{1-F(t)} \big|\Delta\mathbb{\Lambda}_{\dagger}^{(n)}(t) - \Delta\Lambda(t)\big|
+
\sup_{0 \leq t \leq \theta}\big| \mathbb{\Lambda}_{\dagger}^{(n)}(t) - \Lambda(t) \big|.
\end{align*}
According to the first assertion of the theorem, the last term converges almost surely to zero as $n \to \infty$. Concerning the other term, we find that
\begin{align*}
&\sum_{0 \leq t \leq \theta \atop \Delta F(t) > 0} \frac{\Delta F(t)}{1-F(t)} \big|\Delta\mathbb{\Lambda}_{\dagger}^{(n)}(t) - \Delta\Lambda(t)\big| \\
&\leq
\Big(\sup_{0 \leq t \leq \theta}\big| \mathbb{\Lambda}_{\dagger}^{(n)}(t) - \Lambda(t) \big| + \sup_{0 \leq t \leq \theta}\big| \Lambda(t-) - \mathbb{\Lambda}_{\dagger}^{(n)}(t-) \big| \Big)
\sum_{0 \leq t \leq \theta \atop \Delta F(t) > 0} \frac{\Delta F(t)}{1-F(t)} \\
&\leq
\frac{F(\theta)}{1-F(\theta)}
\Big(\sup_{0 \leq t \leq \theta}\big| \mathbb{\Lambda}_{\dagger}^{(n)}(t) - \Lambda(t) \big| + \sup_{0 \leq t \leq \theta}\big| \mathbb{\Lambda}_{\dagger}^{(n)}(t-) - \Lambda(t-) \big| \Big),
\end{align*}
which also converges almost surely to zero as $n \to \infty$, confer with the first assertion of the theorem as well as~\eqref{eq:lambda_crude_left_conv}. Collecting results completes the proof.
\end{proof}

\subsection{Sophisticated expert} \label{sec:sophisticated}

The sophisticated expert provides beliefs in the form of random probability measures $K_i^{(n)}$, $i=1,\dots,n$, on $[0,\infty)$ equipped with the Borel $\sigma$-algebra, such that $K_i^{(n)}\big([0,W_i)\big) =0$. The associated random distribution functions, also denoted $K_i^{(n)}$, $i=1,\ldots,n$, are thus given by
\begin{align*}
K_i^{(n)}(t) = K_i^{(n)}\big([W_i,t)\big),\quad t\geq0.
\end{align*}
Consequently, it holds that $\mathds{1}_{\{W_i > t\}} K_i^{(n)}(t) = 0$ for any $t\geq0$. If the above random distribution functions (almost surely) admit derivatives, then we may express these as
\begin{align*}
k_i^{(n)}(t)=\frac{{\rm d}}{{\rm d}t}K_i^{(n)}(t),\quad \mbox{with} \quad {\int_{W_i}^\infty} k_i^{(n)}(t) \, \mathrm{d}t=1,
\end{align*}
and $k_i^{(n)}$ being a kernel {density} estimate of the whereabouts of $X_i$. To elaborate: Ideally $t \mapsto K_i^{(n)}(t)$ is as close as possible to $t \mapsto \mathds{1}_{\{X_i \leq t\}}$, and so $k_i^{(n)}$ is as close as possible to the Dirac point mass at $X_i$. The nature of such closeness is clarified below. We require that $K_1^{(n)}, \ldots, K_n^{(n)}$ are independent for fixed $n$, but not necessarily identically distributed.

{In the numerical studies, we essentially use (truncated) Gaussian kernels. To be precise, we set
\begin{align*}
k_i(t)
=
\frac{\phi_i(t) \mathds{1}_{\{W_i \leq t\}}}{1-\Phi(W_i)},
\end{align*}
where $\phi_i$ is a Gaussian density and $\Phi_i$ is the corresponding distribution function, both with random parameters depending on, for instance, $W_i$. Popular alternatives to the Gaussian kernel function include the uniform, triangular, Epanechnikov, and sigmoid kernel functions.}

Define the sophisticated expert Kaplan--Meier estimator $\amsmathbb{F}_{\star}^{(n)}$ according to
\begin{align*}
\amsmathbb{F}_{\star}^{(n)}(t) = \frac{1}{n} \sum_{i=1}^n \frac{K_i^{(n)}(t)\delta_i}{1-\amsmathbb{G}^{(n)}(W_i-)}, \quad t\geq0.
\end{align*}
{Due to the condition $K_i^{(n)}\big([0,W_i)\big) =0$, the usual Kaplan--Meier estimator, which corresponds to setting $K_i^{(n)}(t) = \mathds{1}_{\{W_i \leq t\}}$, constitutes an upper bound for $\amsmathbb{F}_{\star}^{(n)}$.}

We also define the oracle Kaplan--Meier estimator $\amsmathbb{F}_{\circ}^{(n)}$ according to
\begin{align*}
\amsmathbb{F}_{\circ}^{(n)}(t) = \frac{1}{n} \sum_{i=1}^n \frac{\mathds{1}_{\{X_i \leq t\}}\delta_i}{1-\amsmathbb{G}^{(n)}(W_i-)}, \quad t\geq0,
\end{align*}
which corresponds to the sophisticated expert Kaplan--Meier estimator with prescience, that is $K_i^{(n)}(t) = \mathds{1}_{\{X_i \leq t\}}$. The oracle Kaplan--Meier estimator is a purely theoretical concept and of course not really relevant in practice.
\begin{lemma}\label{lem:delta_G_link}
Suppose that both contamination and right-censoring are entirely random. Let $0 \leq \theta < H^{-1}(1)$. Then
\begin{align*}
\sup_{0\leq t \leq \theta} \bigg| \frac{1}{n} \sum_{i=1}^n \mathds{1}_{\{X_i \leq t\}} \frac{\delta_i}{1-G(W_i-)} - F(t) \bigg| \overset{\text{a.s.}}{\to} 0, \quad n \to \infty.
\end{align*}
\end{lemma}
\begin{proof}
We first show that
\begin{align}\label{eq:lem_aux}
\amsmathbb{E}\Big[\frac{\delta}{1 - G(W-)} \, \Big| \, X\Big] \overset{\text{a.s.}}{=} 1.
\end{align}
From the local property of the conditional expectation, confer with~Lemma 8.3 in~\cite{Kallenberg2021}, and the fact that
\begin{align*}
\delta = \mathds{1}_{\{Y \leq C, Y < X\}} + \mathds{1}_{\{X \leq C, X \leq Y\}},
\end{align*}
it holds almost surely that
\begin{align*}
\frac{\delta}{1-G(W-)}
&=
\frac{\mathds{1}_{\{Y \leq C, Y < X\}}}{1-G(Y-)}
+
\frac{\mathds{1}_{\{X \leq C, X \leq Y\}}}{1-G(X-)} \\
&=
\frac{\mathds{1}_{\{Y \leq C, Y < X\}}}{\amsmathbb{E}[\mathds{1}_{\{C \geq Y\}} \, | \, Y]}
+
\frac{\mathds{1}_{\{X \leq C, X \leq Y\}}}{\amsmathbb{E}[\mathds{1}_{\{C \geq X\}} \, | \, X]},
\end{align*}
where the last equality uses that $C$ is independent of $(X,Y)$ as per assumption. Employing this representation as well as the fact that $X$ is independent of $(Y,C)$ and that $Y$ and $C$ are conditionally independent given $X$ as per assumption, we find that almost surely
\begin{align*}
\amsmathbb{E}\Big[\frac{\delta}{1 - G(W-)} \, \Big| \, X\Big]
&=
\amsmathbb{E}\Big[\frac{\mathds{1}_{\{Y \leq C, Y < X\}}}{\amsmathbb{E}[\mathds{1}_{\{C \geq Y\}} \, | \, Y]} + \frac{\mathds{1}_{\{X \leq C, X \leq Y\}}}{\amsmathbb{E}[\mathds{1}_{\{C \geq X\}} \, | \, X]} \, \Big| \, X\Big] \\
&=
\amsmathbb{E}[\mathds{1}_{\{Y < X\}}] \frac{\amsmathbb{E}[\mathds{1}_{\{Y \leq C\}}]}{\amsmathbb{E}[\mathds{1}_{\{C \geq Y\}}]}
+
\frac{\amsmathbb{E}[\mathds{1}_{\{X \leq C, X \leq Y\}} \, | \, X]}{\amsmathbb{E}[\mathds{1}_{\{C \geq X\}} \, | \, X]} \\
&=
\amsmathbb{E}[\mathds{1}_{\{Y < X\}}] + \amsmathbb{E}[\mathds{1}_{\{X \leq Y\}}] \frac{\amsmathbb{E}[\mathds{1}_{\{X \leq C\}}]}{\amsmathbb{E}[\mathds{1}_{\{C \geq X\}}]}\\
&= 1,
\end{align*}
which establishes~\eqref{eq:lem_aux}. We now turn our attention to the actual assertion of the lemma. Note that~\eqref{eq:lem_aux} and the law of iterated expectations yield
\begin{align*}
\amsmathbb{E}\Big[\mathds{1}_{\{X \leq t\}}\frac{\delta}{1 - G(W-)} \Big] = F(t), \quad t\geq0,
\end{align*}
so that the strong law of large numbers implies
\begin{align}\label{eq:Ftilde_conv}
\widetilde{\amsmathbb{F}}^{(n)}(t) = \frac{1}{n} \sum_{i=1}^n \mathds{1}_{\{X_i \leq t\}} \frac{\delta_i}{1-G(W_i-)} \overset{\text{a.s.}} \to F(t), \quad n \to \infty,
\end{align}
for any $t\geq0$. {We now extend this pointwise convergence to uniform convergence by a standard argument lifted from the proof of the ordinary Glivenko--Cantelli theorem.} Let $\varepsilon>0$. Fix $M\in\amsmathbb{N}$ such that $\varepsilon^{-1} < M$. Take
\begin{align*}
-\infty = t_0 < t_1 \leq \cdots \leq t_{M-1} < t_M = \infty
\end{align*}
such that
\begin{align*}
F(t_m-) \leq \frac{m}{M} \leq F(t_m), \quad m=1,\ldots,M-1.
\end{align*}
Note that
\begin{align}\label{eq:tilde_help}
F(t_m-) - F(t_{m-1}) < \epsilon \quad \text{for} \quad t_{m-1} < t_m.
\end{align}
Similar to~\eqref{eq:Ftilde_conv}, also
\begin{align*}
\widetilde{\amsmathbb{F}}^{(n)}(t-) \overset{\text{a.s.}} \to F(t-), \quad n \to \infty,
\end{align*}
for any $t\geq0$. In combination, we find that
\begin{align*}
\max_{m=1,\ldots,M-1}\Big\{\big| \widetilde{\amsmathbb{F}}^{(n)}(t_m) - F(t_m)\big|, \big| \widetilde{\amsmathbb{F}}^{(n)}(t_m-) - F(t_m-)\big|\Big\} \overset{\text{a.s.}}{\to} 0, \quad n \to \infty.
\end{align*}
For any $t\geq0$ identify an $m$ such that $t_{m-1} \leq t < t_m$. Then
\begin{align*}
\widetilde{\amsmathbb{F}}^{(n)}(t_{m-1}) - F(t_{m-1}) - \varepsilon \leq \widetilde{\amsmathbb{F}}^{(n)}(t) - F(t) \leq \widetilde{\amsmathbb{F}}^{(n)}(t_m-) - F(t_m-) + \varepsilon
\end{align*}
using~\eqref{eq:tilde_help}. Thus for any $t\geq0$ we find that
\begin{align*}
&\big| \widetilde{\amsmathbb{F}}^{(n)}(t) - F(t) \big| \\
&\leq 
\max_{m=1,\ldots,M-1}\Big\{\big| \widetilde{\amsmathbb{F}}^{(n)}(t_m) - F(t_m)\big|, \big| \widetilde{\amsmathbb{F}}^{(n)}(t_m-) - F(t_m-)\big|\Big\} + \varepsilon
\end{align*}
so that
\begin{align}\label{eq:Ftilde_epsilon}
\begin{split}
&\sup_{0 \leq t < \infty} \big| \widetilde{\amsmathbb{F}}^{(n)}(t) - F(t) \big| \\
&\leq 
\max_{m=1,\ldots,M-1}\Big\{\big| \widetilde{\amsmathbb{F}}^{(n)}(t_m) - F(t_m)\big|, \big| \widetilde{\amsmathbb{F}}^{(n)}(t_m-) - F(t_m-)\big|\Big\} + \varepsilon \\
&\overset{\text{a.s.}}{\to} \varepsilon, \quad n \to \infty.
\end{split}
\end{align}
Let now
\begin{align*}
B_\epsilon = \Big\{ \lim_{n \to \infty} \sup_{0 \leq t < \infty} \big| \widetilde{\amsmathbb{F}}^{(n)}(t) - F(t) \big| \leq \varepsilon\Big\},
\end{align*}
so that~\eqref{eq:Ftilde_epsilon} reads $\amsmathbb{P}(B_\epsilon) = 1$. Downward continuity of measures then yields
\begin{align*}
\amsmathbb{P}\Big(\lim_{n \to \infty} \sup_{0 \leq t < \infty} \big| \widetilde{\amsmathbb{F}}^{(n)}(t) - F(t) \big| = 0\Big)
=
\amsmathbb{P}\Big(\bigcap_{\varepsilon>0} B_\varepsilon\Big) = \lim_{\varepsilon \downarrow 0} \amsmathbb{P}(B_\varepsilon) = 1
\end{align*}
as desired.
\end{proof}
The next result establishes strong consistency on compacts for the sophisticated expert Kaplan--Meier estimator under suitable conditions. We discuss the conditions in more detail in the next subsection.
\begin{theorem}[Strong uniform consistency of the sophisticated expert]\label{thm:cons_soph}
Suppose that both contamination and right-censoring are entirely random. Let $0\leq\theta<H^{-1}(1)$. Suppose that
\begin{align}\label{eq:sophisticated_assumption}
\amsmathbb{E}\Big[ \frac{1}{n}\sum_{i=1}^n  \sup_{0 \leq t \leq \theta} \big| K_i^{(n)}(t) - \mathds{1}_{\{X_i \leq t\}} \big|\Big] \to 0, \quad n \to \infty.
\end{align}
Then
\begin{align*}
\sup_{0 \leq t \leq \theta}\big| \amsmathbb{F}_{\star}^{(n)}(t) - F(t) \big| \overset{\text{a.s.}}{\to} 0, \quad n \to \infty.
\end{align*}
\end{theorem}
\begin{proof}
Note that
\begin{align}\label{eq:sophisticated_split}
\begin{split}
&\sup_{0 \leq t \leq \theta}\big| \amsmathbb{F}_{\star}^{(n)}(t) - F(t) \big| \\
&\leq
\sup_{0 \leq t \leq \theta} \bigg| \frac{1}{n} \sum_{i=1}^n \mathds{1}_{\{X_i \leq t\}} \frac{\delta_i}{1-G(W_i-)} - F(t) \bigg| \\
&\quad+
\sup_{0 \leq t \leq \theta} \bigg| \frac{1}{n} \sum_{i=1}^n \big(K_i^{(n)}(t)-\mathds{1}_{\{X_i \leq t\}}\big) \frac{\delta_i}{1-G(W_i-)} \bigg| \\
&\quad+
\sup_{0 \leq t \leq \theta} \bigg| \frac{1}{n} \sum_{i=1}^n K_i^{(n)}(t) \Big(\frac{\delta_i}{1-\amsmathbb{G}^{(n)}(W_i-)} - \frac{\delta_i}{1-G(W_i-)} \Big) \bigg|.
\end{split}
\end{align}
The right-hand side of the first line converges almost surely to zero as $n\to\infty$ according to Lemma~\ref{lem:delta_G_link}. Turning our attention to the second line and using that both $\mathds{1}_{\{W_i>t\}}K_i^{(n)}(t)=0$ and $\mathds{1}_{\{W_i > t\}}\mathds{1}_{\{X_i \leq t\}} = 0$ for any $t\geq0$, we find that
\begin{align*}
&\sup_{0 \leq t \leq \theta} \bigg| \frac{1}{n} \sum_{i=1}^n \big(K_i^{(n)}(t)-\mathds{1}_{\{X_i \leq t\}}\big) \frac{\delta_i}{1-G(W_i-)} \bigg| \\
&\leq
\frac{1}{1-G(\theta-)}\frac{1}{n} \sum_{i=1}^n \sup_{0 \leq t \leq \theta} \big| K_i^{(n)}(t) - \mathds{1}_{\{X_i \leq t\}} \big| \\
&\leq \frac{1}{1-G(\theta-)}\bigg|\frac{1}{n} \sum_{i=1}^n  \sup_{0 \leq t \leq \theta} \big| K_i^{(n)}(t) - \mathds{1}_{\{X_i \leq t\}} \big| - \amsmathbb{E}\big[\sup_{0 \leq t \leq \theta} \big| K_i^{(n)}(t) - \mathds{1}_{\{X_i \leq t\}} \big| \big] \bigg| \\
&\quad+ \frac{1}{1-G(\theta-)} \amsmathbb{E}\Big[\frac{1}{n}\sum_{i=1}^n\sup_{0 \leq t \leq \theta} \big| K_i^{(n)}(t) - \mathds{1}_{\{X_i \leq t\}} \big| \Big],
\end{align*}
The last line converges almost surely to zero as $n\to\infty$; this is by assumption. Furthermore, since 
\begin{align*}
\sup_{0 \leq t \leq \theta} \big| K_i^{(n)}(t) - \mathds{1}_{\{X_i \leq t\}} \big| \leq 1,
\end{align*}
the second to last line likewise converges almost surely to zero according to the strong law of large numbers for triangular arrays. We conclude that the second line of~\eqref{eq:sophisticated_split} converges almost surely to zero. Furthermore,
\begin{align*}
&\sup_{0 \leq t \leq \theta} \bigg| \frac{1}{n} \sum_{i=1}^n K_i^{(n)}(t) \Big(\frac{\delta_i}{1-\amsmathbb{G}^{(n)}(W_i-)} - \frac{\delta_i}{1-G(W_i-)} \Big) \bigg| \\
&\leq
\frac{1}{n} \sum_{i=1}^n \sup_{0 \leq t \leq \theta} K_i^{(n)}(t) \delta_i \Big| \frac{1}{1-\amsmathbb{G}^{(n)}(W_i-)} - \frac{1}{1-G(W_i-)} \Big| \\
&\leq
\sup_{0 \leq t \leq \theta} \Big| \frac{1}{1-\amsmathbb{G}^{(n)}(t-)} - \frac{1}{1-G(t-)} \Big|
\end{align*}
where we have utilized that $0\leq K_i^{(n)}\leq1$ and, yet again, that $\mathds{1}_{\{W_i>t\}}K_i^{(n)}(t)=0$ for any $t\geq0$. Since $\amsmathbb{G}^{(n)}$ and $G$ are càdlàg, we find from~\eqref{eq:km_censored_cons} that
\begin{align*}
\sup_{0\leq t \leq \theta} \big| \amsmathbb{G}^{(n)}(t-) - G(t-) \big| \overset{\text{a.s.}}{\to} 0, \quad n \to \infty.
\end{align*}
Note that $[0,G(\theta-)] \ni x \mapsto \frac{1}{1-x}$ is uniformly continuous. Thus
\begin{align*}
\sup_{0 \leq t \leq \theta} \Big| \frac{1}{1-\amsmathbb{G}^{(n)}(t-)} - \frac{1}{1-G(t-)} \Big| \overset{\text{a.s.}}{\to} 0, \quad n \to \infty.
\end{align*}
All in all, the third and final line of~\eqref{eq:sophisticated_split} converges almost surely to zero. Collecting results then completes the proof.
\end{proof}

\begin{remark}
If the expert beliefs are identically distributed, then~\eqref{eq:sophisticated_assumption} reads
\begin{align*}
\amsmathbb{E}\Big[ \sup_{0 \leq t \leq \theta} \big| K_1^{(n)}(t) - \mathds{1}_{\{X_1 \leq t\}} \big|\Big] \to 0, \quad n \to \infty.
\end{align*}
\end{remark}

\begin{remark}
A sufficient but not necessary condition for~\eqref{eq:sophisticated_assumption} to hold is obviously
\begin{align*}
\sup_{1\leq i\leq n}\amsmathbb{E}\Big[ \sup_{0 \leq t \leq \theta} \big| K_i^{(n)}(t) - \mathds{1}_{\{X_i \leq t\}} \big|\Big] \to 0, \quad n \to \infty.
\end{align*}
\end{remark}

\begin{corollary}
Suppose that both contamination and right-censoring are entirely random. Let $0\leq\theta<H^{-1}(1)$. Then
\begin{align*}
\sup_{0 \leq t \leq \theta}\big| \amsmathbb{F}_{\circ}^{(n)}(t) - F(t) \big| \overset{\text{a.s.}}{\to} 0, \quad n \to \infty.
\end{align*}
\end{corollary}

\subsection{Comparison} \label{sec:compare}
{
There are some fundamental differences between the crude and the sophisticated expert Kaplan--Meier estimators. The crude expert Kaplan--Meier estimator is essentially the usual Kaplan--Meier estimator, but using a different set of observations.} Consequently, the representation of the Kaplan--Meier estimator as an inverse probability of censoring weighted average is not important in the formulation of the crude expert Kaplan--Meier estimator. This is different from the sophisticated expert Kaplan--Meier estimator, for which this representation is absolutely fundamental. Theorem~\ref{thm:cons_crude} and Theorem~\ref{thm:cons_soph} establish strong consistency on compacts of the crude and sophisticated expert Kaplan--Meier estimators, respectively, but subject to quite different conditions. First, for the crude expert Kaplan--Meier estimator we only need to assume that contaminated right-censoring is entirely random, while in case of the sophisticated counterpart, we need the stronger condition that both contamination and right-censoring are entirely random. Secondly, consistency of the sophisticated estimator requires a rather strong condition on the quality of expert information compared to the crude estimator. For the crude expert Kaplan--Meier estimator, the convergence in mean given by~\eqref{eq:crude_assumption} is sufficient. For the sophisticated expert Kaplan--Meier estimator, we instead rely on the $L^{1}$-convergence given by~\eqref{eq:sophisticated_assumption}. It is of course not entirely surprising that a more advanced estimator requires stronger conditions. {This is also explored in the numerical example of Section~\ref{sec:num}.}

\section{Extensions}\label{sec:ext}

{
In this section, we further delineate on research directions of both theoretical and practical interest. We first consider the rates of convergence of the proposed estimators, though this is primarily of mathematical interest, since further detailed assumptions on the asymptotic quality of expert information, which are necessary for such analyses, are unverifiable in most practical settings. This is followed by discussions on pre-specified parametric tail behavior and covariates, which are common features in actuarial science, finance, hydrology, and other areas of application (see, for instance,~\cite{Beirlant2004,Embrechts2013}).

\subsection{Rates of convergence}\label{sec:cov}
In Section~\ref{sec:expert}, we introduced the crude and expert Kaplan--Meier estimators, which under certain assumptions on the asymptotic quality of expert information, namely assumptions~\eqref{eq:crude_assumption} and~\eqref{eq:sophisticated_assumption}, can be shown to strongly and uniformly converge to the true underlying distribution. 

If the rate of convergences of Theorems~\ref{thm:cons_crude} and~\ref{thm:cons_soph} are of interest, more detailed assumptions on the expert information are required. In the following, a sketch of one possible approach for the crude expert is provided. We see along the lines of the proof of Lemma~\ref{lem:crude_expert} that obtaining a weak convergence theorem for the process
\begin{align*}
 \left\{ n^{1/2}\left(\amsmathbb{H}_{\dagger,1}^{(n)}(t) - H_1(t)\right) \right\}_{0\le t\le\theta}
\end{align*}
is possible under the following assumption of a faster disappearance of the corresponding bias term:
\begin{align*}
\sup_{0 \leq t < \infty} \bigg|
n^{-1/2}
\sum_{i=1}^n \amsmathbb{E}\big[\mathds{1}_{\{W_i \leq t\}} \eta_i^{(n)}\big] - n^{1/2}H_1(t)\bigg| \to 0, \quad n \to \infty.
\end{align*}
A similar weak convergence result for the process $\{n^{1/2}(\amsmathbb{H}^{(n)}(t)-H(t))\}_{t\ge0}$ is standard, and it is not too difficult to see that the two processes then are asymptotically jointly Gaussian. Since both $\mathbb{\Lambda}_{\dagger}^{(n)}$ and $\amsmathbb{F}_{\dagger}^{(n)}$ can be written as Hadamard-differentiable operators of the pair $(\amsmathbb{H}_{\dagger,1}^{(n)},\amsmathbb{H}^{(n)})$, we then obtain from the functional delta method (confer with Theorem~20.8 in~\cite{Vaart1998}) weak convergence of the following processes:
\begin{align*}
\left\{n^{1/2}\left(\mathbb{\Lambda}_{\dagger}^{(n)}(t) - \Lambda(t)\right)\right\}_{0\le t\le\theta},\quad \left\{n^{1/2}\left(\amsmathbb{F}_\dagger^{(n)}(t) - F(t)\right) \right\}_{0\le t\le\theta},
\end{align*}
the limits being Gaussian. In other words, we may, under a suitable assumption on the asymptotic quality of expert information, show that the rate of convergence for the crude expert estimator is of order $O_{\text{P}}(n^{-1/2})$.

On the other hand, for the sophisticated expert, we may not directly make use of the functional delta method, the main difficulty being that the kernel functions $K_i$ are specific to the individual, so we cannot write $\amsmathbb{F}_{\star}^{(n)}$ as a functional of simpler components. The same difficulty prevents us from tackling the problem directly through the theory of Kaplan--Meier integrals, which would be the case if all the kernel functions were identical, say $K_i = K$, see also~\cite{Stute1995}. 

Finally, we would like to remark that the rate of convergence for the sophisticated expert estimator will depend on some bandwidth sequence $(a_n)$, since we should require the kernel functions to degenerate into Dirac measures in the limit at a certain rate. From~\cite{Dabrowska1987}, we rationalize that the rate is likely of order $O_{\text{P}}\big((na_n)^{-1/2}\big)$. Further investigation is outside of the scope of this paper but an interesting avenue for further research.}

\subsection{Towards semi-parametric estimation}

For certain applications, fully non-parametric procedures can have some shortcomings. For instance, the tail behavior of $\amsmathbb{F}_{\star}^{(n)}$ will be determined by that of the kernels $K_i^{(n)}$, which have been constructed for the related but different task of producing an expert judgment on $X_i$. Perhaps more prominently, the tail behavior of $\amsmathbb{F}_{\dagger}^{(n)}$ is left unspecified, as is already the case for the usual Kaplan--Meier estimator. We now aim to amend our estimators to allow for parametric approximations.

To set some notation, we assume that
\begin{align*}
F\in\mathcal{D}=\{F_\theta,\, \theta \in \Theta\},
\end{align*}
with $\Theta$ being a non-empty parameter space. Let $F_1$ and $F_2$ be two absolutely continuous distribution functions on $[0,\infty)$, such that we may define their Kullback--Leibler divergence by
\begin{align*}
D_{\text{KL}}(F_1||F_2)=\int_{[0,\infty)} \log(f_1(t)/f_2(t)) \, F_1(\mathrm{d}t).
\end{align*}
If $F_1\in \mathcal{D}$, say $F_1=F_{\theta_1}$, we may actually recover its parameter $\theta_1$ by minimizing
\begin{align*}
\theta_1=\arg\min_{\theta\in\Theta}D_{\text{KL}}(F_1||F_\theta)=\arg\max_{\theta\in\Theta}\int_{[0,\infty)} \log{\big(f_\theta(t)\big)} \, F_1(\mathrm{d}t).
\end{align*}
The last expression is particularly useful in practice since it allows us to make the following colloquial remark.
\begin{remark}
Let $\amsmathbb{X}^{(n)}$ denote the empirical distribution function of a fully observed sample $X_1, \ldots, X_n$. Then 
\begin{align*}
\argmin_{\theta\in\Theta}D_{\text{KL}}(\amsmathbb{X}^{(n)}||F_\theta)
&\overset{\ast}{=}
\argmax_{\theta\in\Theta}\int_{[0,\infty)} \log\big(f_\theta(t)\big) \, \amsmathbb{X}^{(n)}(\mathrm{d}t)\ \\
&=
\argmax_{\theta\in\Theta}\sum \log\big(f_\theta(X_i)\big),
\end{align*}
such that the minimizer of the Kullback--Leibler divergence between $\amsmathbb{X}^{(n)}$ and $\mathcal{D}$ is simply the maximum likelihood estimator of the sample. Notice that we have placed an asterix above the first equality. This is due to the fact that $\amsmathbb{X}_n$ does not admit a density with respect to the Lebesgue measure. However, the latter fact is inconsequential for the minimization problem itself, so we may study the minimization $D_{\text{KL}}(\amsmathbb{X}_n||F_\theta)$ in such an extended sense.
\end{remark}
Returning to our setting, the basic idea in order to incorporate expert judgment into the estimation procedure is to first capture the information by constructing either $\amsmathbb{F}_{\dagger}^{(n)}$ or $\amsmathbb{F}_{\star}^{(n)}$ and then to minimize the Kullback--Leibler divergence between such an estimator and $\mathcal{D}$. In case the estimator does not possess a density, an extended minimization problem is to be understood as in the above remark. Thus, we may define
\begin{align*}
\theta_\dagger&=\arg\min_{\theta\in\Theta}D_{\text{KL}}(\amsmathbb{F}_{\dagger}^{(n)}||F_\theta)\\
&=\arg\max_{\theta\in\Theta}\int \log\big(f_\theta(t)\big) \, \amsmathbb{F}_{\dagger}^{(n)}(\mathrm{d}t)\\
&=\arg\max_{\theta\in\Theta}\frac{1}{n}  \sum_{i=1}^n\log\big(f_\theta(W_i)\big) \frac{ \eta_i^{(n)}}{1 - \amsmathbb{L}'\big(W_i-;(W_i,1-\eta_i^{(n)})_{i=1}^n\big)}
\end{align*}
and, similarly,
\begin{align*}
\theta_\star&=\arg\min_{\theta\in\Theta}D_{\text{KL}}(\amsmathbb{F}_{\star}^{(n)}||F_\theta)\\
&=\arg\max_{\theta\in\Theta}\int \log\big(f_\theta(t)\big)\,\amsmathbb{F}_{\star}^{(n)}(\mathrm{d}t)\\
&=\arg\max_{\theta\in\Theta} \frac{1}{n} \sum_{i=1}^n {\int_{[W_i,\infty)}} \log\big(f_\theta(t)\big) \,K_i^{(n)}(\mathrm{d}t) \frac{\delta_i}{1-\amsmathbb{G}^{(n)}(W_i-)}. 
\end{align*}
{For integrands not depending on $n$, these types of estimators are well-understood and have been studied in various settings, see for instance~\cite{Zhou1992,StuteWang1993,Stute1999,Lopez2009}. More recently, and in connection with expert information, the theory of informed censoring outlined in~\cite{Bladt2022} directly builds on the theory of M-type estimators of~\cite{Vaart1998} to offer a parametric solution.} Adapting the aforementioned results to encompass arrays of M-type estimators with random measures remains a promising line of further research.

\begin{example}
Let $\mathcal{D}$ consist of Exponential laws with densities
\begin{align*}
f_\lambda(t)=\lambda \exp(-\lambda t),\quad t,\lambda>0.
\end{align*}
The crude semi-parametric expert estimator is obtained as
\begin{align*}
&\lambda_\dagger\\
&=
\left\{\sum_{i=1}^n 
\frac{ \eta_i^{(n)}}{1 - \amsmathbb{L}'\big(W_i-;(W_i,1-\eta_i^{(n)})_{i=1}^n\big)}
\right\}
\Big/\left\{\sum_{i=1}^n
\frac{W_i \eta_i^{(n)}}{1 - \amsmathbb{L}'\big(W_i-;(W_i,1-\eta_i^{(n)})_{i=1}^n\big)}
\right\}\\
&=\left\{n(1-\amsmathbb{F}_\dagger^{(n)}(W_{n:n}))\right\}\Big/\left\{\sum_{i=1}^n
\frac{W_i \eta_i^{(n)}}{1 - \amsmathbb{L}'\big(W_i-;(W_i,1-\eta_i^{(n)})_{i=1}^n\big)}
\right\}\!.
\end{align*}
Now let $k_i^{(n)}$ be a left-truncated (at $W_i$) Gamma density with parameters $\alpha_i^{(n)}$ and $\beta_i^{(n)}$. Then we obtain
\begin{align*}
\int_{W_i}^\infty \log(f_\lambda(t))k_i^{(n)}(t)\,\mathrm{d} t&=
\int_{W_i}^\infty [\log(\lambda)-\lambda t] \frac{{\beta_i^{(n)}}^{\alpha_i^{(n)}} {t}^{\alpha_i^{(n)}-1} \exp(-\beta_i^{(n)}t)  }{\overline\gamma(\alpha_i^{(n)},\beta_i^{(n)}W_i)}\,\mathrm{d} t\\
&=\log(\lambda)-\frac{\lambda}{\beta_i^{(n)}}\frac{\overline\gamma(\alpha_i^{(n)}+1,\beta_i^{(n)}W_i)}{\overline\gamma(\alpha_i^{(n)},\beta_i^{(n)}W_i)},
\end{align*}
where $\overline{\gamma}(s,x)=\int_x^\infty t^{s-1}e^{-t}\mathrm{d}t$ is the upper incomplete gamma function. It then follows that the sophisticated semi-parametric expert estimator is given by
\begin{align*}
\lambda_\star&= \left\{\sum_{i=1}^n\frac{\delta_i}{1-\amsmathbb{G}^{(n)}(W_i-)}\right\}
\Big/\left\{\sum_{i=1}^n
\frac{\delta_i}{1-\amsmathbb{G}^{(n)}(W_i-)}
\frac{\overline\gamma(\alpha_i^{(n)}+1,\beta_i^{(n)}W_i)}{\beta_i^{(n)}\overline\gamma(\alpha_i^{(n)},\beta_i^{(n)}W_i)}\right\}\\
&=\left\{n(1-\amsmathbb{F}^{(n)}(W_{n:n}))\right\}
\Big/\left\{\sum_{i=1}^n
\frac{\delta_i}{1-\amsmathbb{G}^{(n)}(W_i-)}
\frac{\overline\gamma(\alpha_i^{(n)}+1,\beta_i^{(n)}W_i)}{\beta_i^{(n)}\overline\gamma(\alpha_i^{(n)},\beta_i^{(n)}W_i)}\right\}\!.
\end{align*}
\end{example}

\begin{example}
Let $\mathcal{D}$ consist of Pareto laws with densities
$$f_\alpha(t)=\frac\alpha\sigma (t/\sigma)^{-\alpha-1},\quad t>\sigma>0,\quad\alpha>0.$$
Suppose first that $\sigma$ is known. The crude semi-parametric expert estimator of $\alpha$ is obtained as
\begin{align*}
\alpha_\dagger=
\left\{n(1-\amsmathbb{F}_\dagger^{(n)}(W_{n:n}))\right\}
\Big/\left\{\sum_{i=1}^n
\frac{\log(W_i/\sigma) \eta_i^{(n)}}{1 - \amsmathbb{L}'\big(W_i-;(W_i,1-\eta_i^{(n)})_{i=1}^n\big)}
\right\}\!.
\end{align*}
Similarly, for the sophisticated expert we obtain
\begin{align*}
\alpha_\star= \left\{n(1-\amsmathbb{F}^{(n)}(W_{n:n}))\right\}
\Big/\left\{\sum_{i=1}^n
\frac{\delta_i}{1-\amsmathbb{G}^{(n)}(W_i-)}
\int_{W_i}^\infty \log(t/\sigma) \, K_i^{(n)}(\mathrm{d}t) \right\}\!.
\end{align*}
In the context of extreme value theory, one may replace the lower limit $\sigma$ with a suitable large order statistic and consider observations above such a threshold. We then obtain Hill-type estimators (confer with~\cite{Hill1975}) of the form
\begin{align*}
\alpha^k_\dagger&=
\left\{n(1-\amsmathbb{F}_\dagger^{(n)}(W_{n-k:n}))\right\}
\Big/\left\{\sum_{i=n-k+1}^n
\frac{\log(W_{i:n}/W_{n-k:n}) \eta_{i:n}^{(n)}}{1 - \amsmathbb{L}'\big(W_{i:n}-;(W_i,1-\eta_i^{(n)})_{i=1}^n\big)}
\right\}\\
\alpha^k_\star&= \left\{n(1-\amsmathbb{F}^{(n)}(W_{n-k:n}))\right\}
\Big/ \\ &\quad\left\{\sum_{i=n-k+1}^n
\frac{\delta_i}{1-\amsmathbb{G}^{(n)}(W_{i:n}-)}
\int_{W_{i:n}}^\infty \log(t/W_{n-k:n}) \, K_{i:n}^{(n)}(\mathrm{d}t) \right\},
\end{align*}
for $k=1,\dots,n-1$. The selection of the optimal $k$, as well as the derivation of the asymptotic distribution of the second extreme value estimator, is no trivial task and, although highly interesting, is not pursued here. In the absence of contamination, the first estimator has been well-studied in the literature, see for instance~\cite{WormsWorms2014}. The second, in the absence of contamination, is a new smoothed version of the estimator in~\cite{WormsWorms2014}.
\end{example}

In essence, as the above examples show, the semi-parametric approach can yield quantities which generalize or smooth well-known estimators, making it a valuable method beyond our fully non-parametric approach.

\subsection{Covariates}\label{sec:cov}

For many practical applications, the inclusion of covariates is of great importance. We should therefore like to briefly address this matter in relation to the work presented here.

First, it should be noted that the expert judgments and beliefs may actually depend on any covariates since the expert can use different distributions between claims. This, in particular, allows for the use of internal covariates, such as -- in the context of disability insurance -- whether claim closure is the result of either preliminary recovery or instead death; claims that are closed due to death typically do not reopen and are thus rarely falsely closed.

Second, in regards to the distribution of interest, namely the distribution of the size or duration of an insurance claim $X$, we disregard the effect of covariates. To include covariates, one starting point could be to develop conditional expert Kaplan--Meier estimators by including expert judgments and beliefs in the ordinary conditional Kaplan--Meier estimators. The latter estimators were introduced in~\cite{Beran1981} and further studied in for instance~\cite{Dabrowska1989}. 

\section{Numerical study}\label{sec:num}

In this concluding section, we illustrate the applicability and practical connotations of expert Kaplan--Meier estimators. We focus on applications to disability insurance, using both simulated as well as real data from a French insurer. The former is the focal point of Subsection~\ref{sec:sim}, while Subsection~\ref{sec:real} is devoted to the latter. 

\subsection{Simulated dataset}\label{sec:sim}

This section is a study based on simulated real-world scenarios from the health and disability sector and serves a dual purpose. First, the study validates the consistency and overall quality of the estimators presented and, secondly, we have selected a realistic situation and parameters, so practical applications could follow closely from this study. En passant, we suggest multiple ways in which expert information may be formed. Due to the obviously sensitive nature of expert information on the lives and health of policyholders, to the best of the authors' knowledge, this component is not publicly available.

We are interested in learning the recovery rates for policyholders claiming disability benefits. In the analysis, we take as the basic time dimension the duration since onset of disability, so that staggered entry results in random right-censoring. The selection of appropriate time scales is outside the scope of this study, but further details may for instance be found in~\cite{AndersenBorganGillKeiding1993}. It is worthwhile to remark that we carry out the analysis for parameters that would correspond to policyholders of a fixed age of $40$ years at the onset of disability. This fixing of age will also be encountered in the real-world dataset below, and though the introduction of age as a covariate is of interest, it is not englobed by our models and thus not studied further; confer also with Subsection~\ref{sec:cov}.

Both contamination and right-censoring are assumed to be entirely random, that is, $X$, $Y$, and $C$ are mutually independent. We moreover introduce some further nomenclature: the hazard rates for $X$ and $Y$ are denoted by $\mu_{01}$ and $\mu_{02}$, respectively. Thus, the hazard rate for $X \wedge Y$ is simply $\mu_{01} + \mu_{02}$; see Figure~\ref{fig:multiple_cause} for a depiction of this construction as a multi-state model.

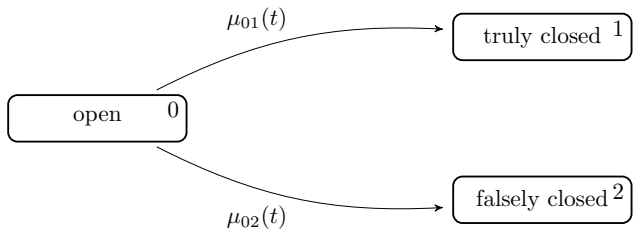
\begin{figure}[h]
	\centering
	\scalebox{0.88}{
	\begin{tikzpicture}[node distance=2em and 0em]
		\node[punktl, draw = none] (2) {};
		\node[punktl, left = 40mm of 2] (0) {open};
		\node[anchor=north east, at=(0.north east)]{$0$};
		\node[punktl, below = 5mm of 2] (3) {falsely closed};
		\node[anchor=north east, at=(3.north east)]{$2$};
		\node[punktl, above = 5mm of 2] (1) {truly closed};
		\node[anchor=north east, at=(1.north east)]{$1$};
	\path
		(0)	edge [pil, bend left = 15]		node [above left]		{$\mu_{01}(t)$}				(1)
		(0)	edge [pil, bend right = 15]		node [below left]		{$\mu_{02}(t)$}				(3)
	;
	\end{tikzpicture}}
	\caption{The multi-state model underlying the simulated disability duration dataset.}
	\label{fig:multiple_cause}
\end{figure}

As previously mentioned, we intend to keep parameters as realistic as possible. For this reason, we choose $C \sim \mbox{Uniform}([0,20])$, which is compatible with the characteristics of an insurer that has been in business for about two decades and has experienced limited fluctuations in the number of disability claims. Furthermore, we select the following hazard rates for $X$ and $Y$:
\begin{align*}
[0,20] \ni t &\mapsto \mu_{01}(t)+\mu_{02}(t)=\exp(0.1-1.5\,t), \\
[0,20] \ni t &\mapsto \mu_{02}(t)=  \exp(0.1-1.5\,t)/8 + \exp(0.1-2.5\,t)/4,
\end{align*}
see also the bottom left panel of Figure~\ref{fig:new_sim_data} for a plot of the hazard rates. From our experience, this specification produces data that is comparable to practical actuarial situations.

With $p$ as in~\eqref{eq:p_function}, it holds for $w \in [0,20]$ that
\begin{align*}
p(w) = \frac{\mu_{01}(w)}{\mu_{01}(w)+\mu_{02}(w)},
\end{align*}
which is depicted in the bottom right panel of Figure~\ref{fig:new_sim_data}. In terms of the multi-state model, this essentially yields a conditional mark distribution.

\begin{figure}[!htbp]
\centering
\includegraphics[width=0.44\textwidth]{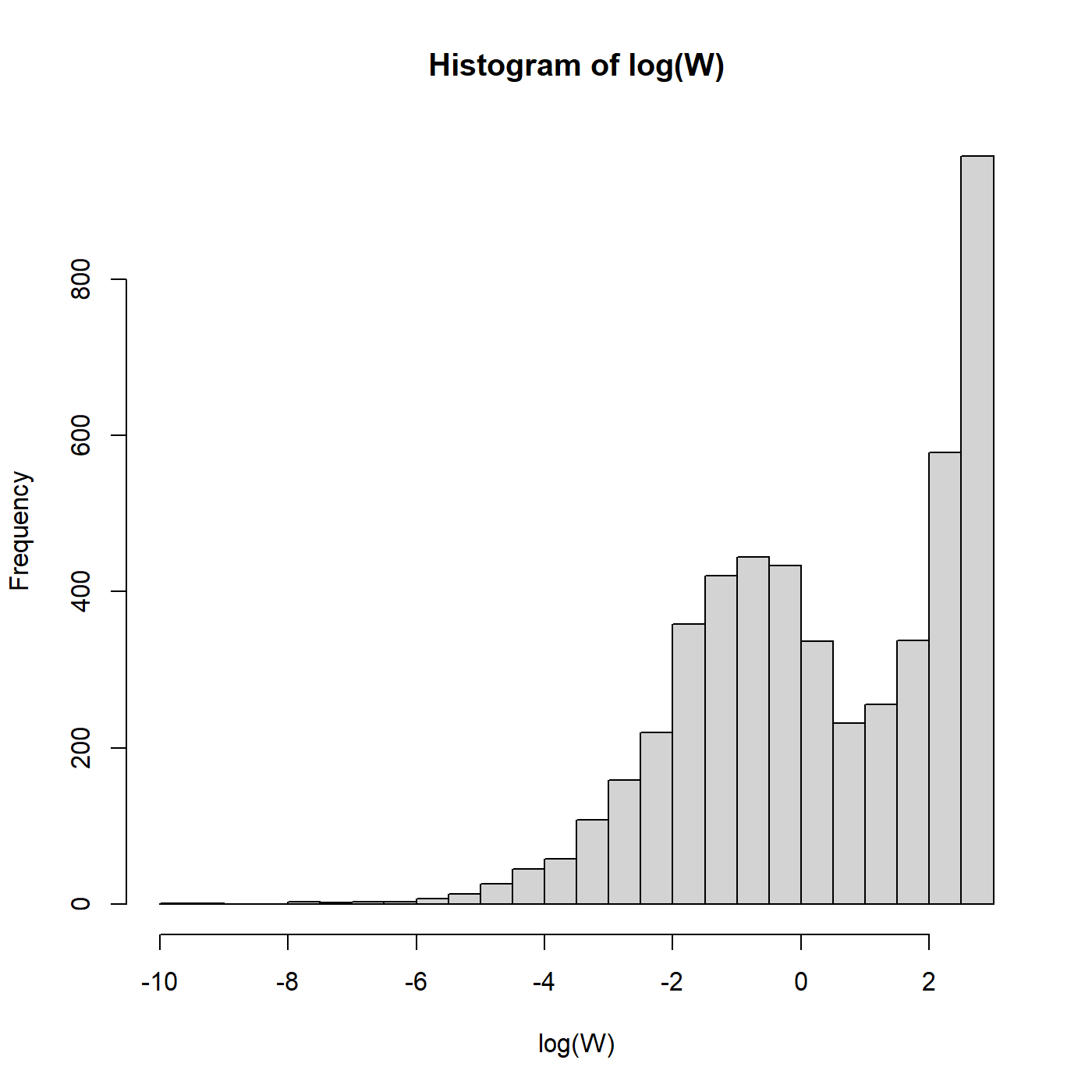}
\includegraphics[width=0.44\textwidth,trim= 0in -1.4in 0in 0in,clip]{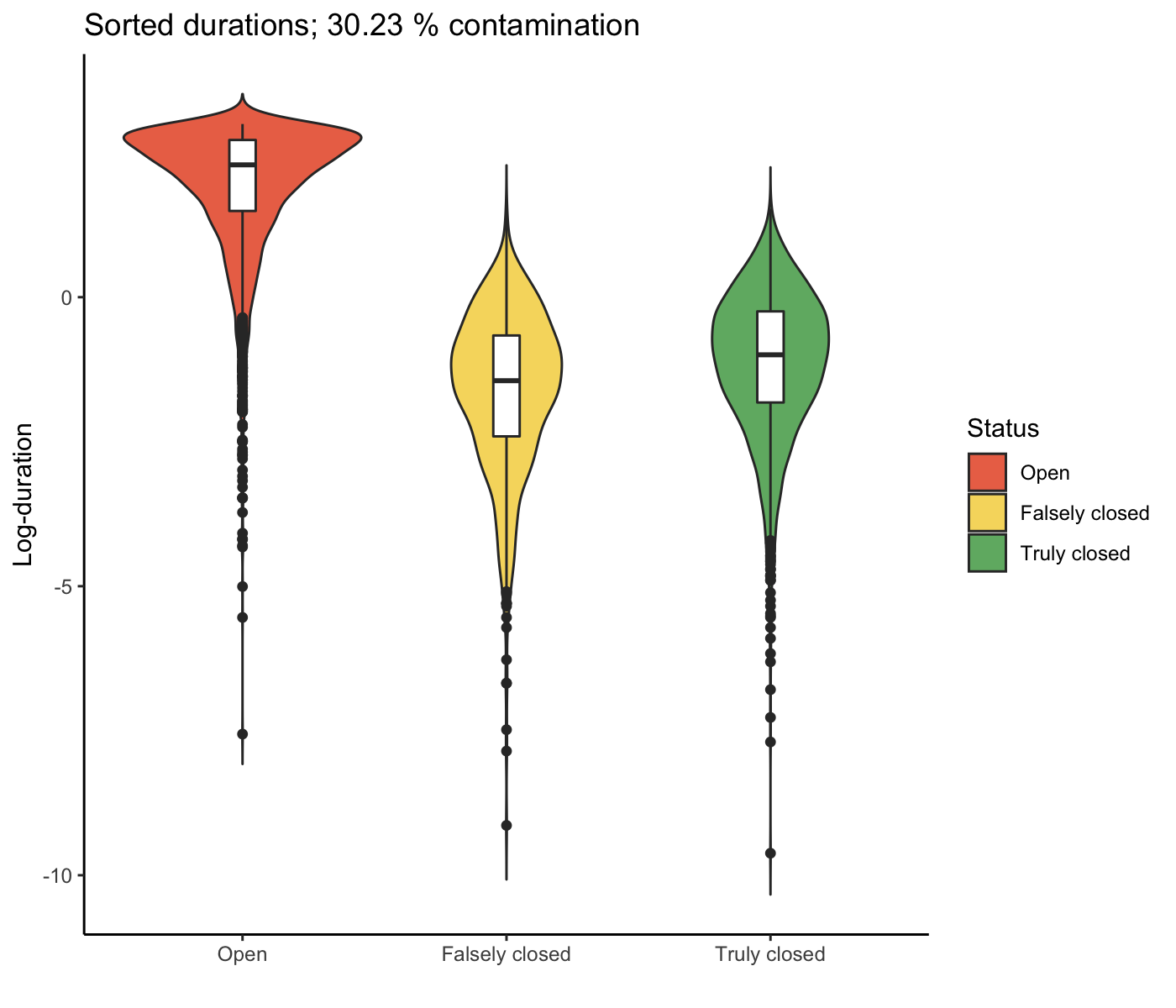}
\includegraphics[width=0.44\textwidth]{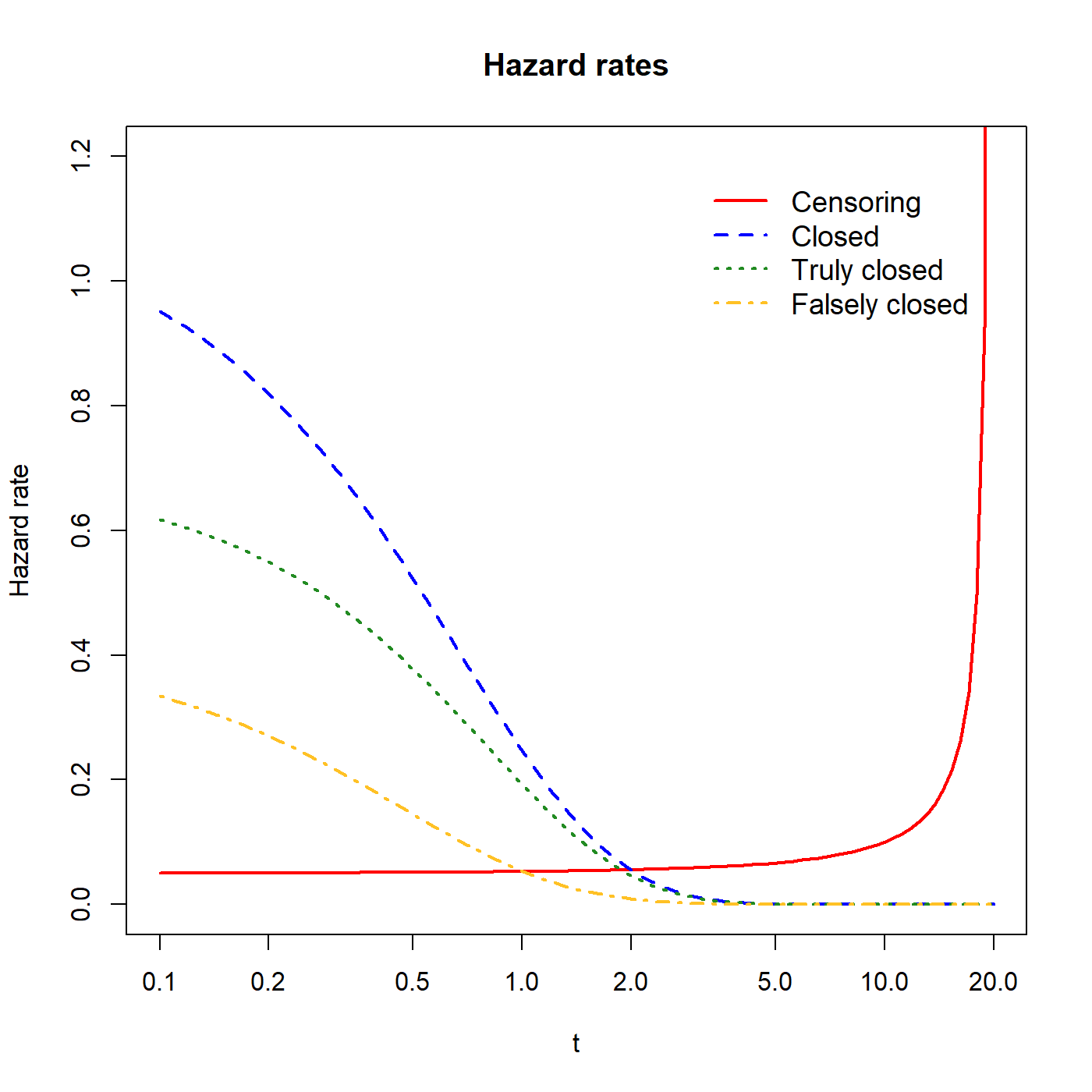}
\includegraphics[width=0.44\textwidth]{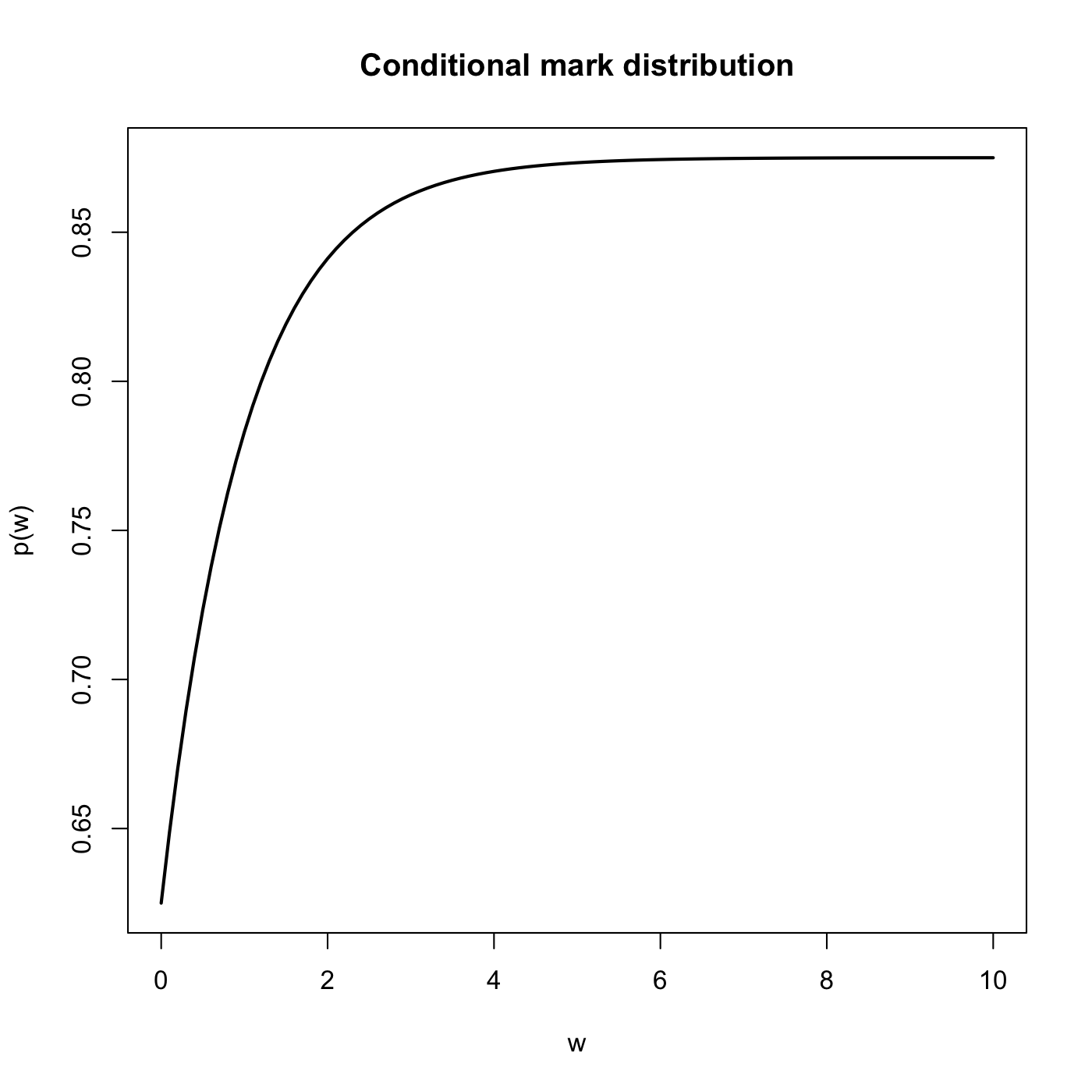}
\caption{The simulated disability duration dataset. Top panels: histogram of log-durations (left) and a violin plot of log-durations colored according to claim status (right). Bottom panels: the hazard rates of $X$, $Y$, and $C$ (left) and the conditional mark distribution $w\mapsto p(w)$ (right).
}
\label{fig:new_sim_data}
\end{figure}

We simulate $5{,}000$ iid pairs $(W_k,\delta_k)_{k=1}^{5000}$ according to the above specification and obtain that roughly $30.23\,\%$ of the observations are contaminated. It should be noted that the proportion of contamination is not computable from the observed data. The top two panels in Figure~\ref{fig:new_sim_data} help visualize the data.

Regarding the generation of expert information, we consider two scenarios for each of the two experts (crude and sophisticated). In the following, we describe the procedure in detail. Note that key components are depicted it in Figure~\ref{fig:new_sim_expert_info}.

For the crude expert, we consider two cases both generated using the same mechanism but using different parameters. We simulate iid random variables $(B_k)_{k=1}^{5000}$ according to
\begin{align*}
B_k \, | \, W_k \sim \mbox{Bernoulli}\big( p_0\cdot p(W_k)+(1-p_0)\big)
\end{align*}
and set $\eta_k=\delta_k\cdot B_k$ with $p_0 = 0.75$ for the first crude expert and $p_0=0.95$ for the second crude expert. {We stress that in our theoretical framework, the expert only provides the judgments $(\eta_k)_{k=1}^{5000}$ and no information whatsoever about how these judgements are formed, such as the expert's specification (or estimation) of $p_0$ and $p$.} Notice that the effectiveness of the expert does not require them to have access to the contamination mechanism directly, but rather indirectly through the knowledge of $w \mapsto p(w)$. Taking the conditional expectation of $B_k$, we may define the \textit{expert weighted conditional mark function} as
\begin{align}\label{eq:cond_meanmark}
(w,p_0) \mapsto \amsmathbb{E}[B_1 \, | \, W_1 = w]=p_0\cdot p(w) + (1-p_0).
\end{align}
{The interpretation is that} the expert effectiveness $p_0$ bridges between perfect information corresponding to knowledge of the conditional mark distribution $w \mapsto p(w)$ and no information corresponding to $w \mapsto 1$. The top left panel of Figure~\ref{fig:new_sim_expert_info} shows the contours of~\eqref{eq:cond_meanmark}, while the top right panel shows the functions for the selected expert effectivenesses $p_0=0.75, 0.95$.

For the sophisticated expert, we assume some stronger form of information, namely access to noisy proxies for the unobservable $(X_k)_{k=1}^{5000}$. Concretely, for the first expert, we impose {truncated} Gaussian kernels with mean $X_k+V^{(1)}_k$ and standard deviation $X_k+V^{(2)}_k$, where $V^{(1)}_k\stackrel{iid}{\sim}\Gamma(1,1)$ and  $V^{(2)}_k\stackrel{iid}{\sim}\Gamma(1,10)$. The second expert is specified similarly but the noise random variables have distributions $\Gamma(10,10)$ and $\Gamma(1,100)$, respectively. The resulting {means and variances of the kernels} are given in the bottom panel of Figure~\ref{fig:new_sim_expert_info}.

\begin{figure}[!htbp]
\centering
\includegraphics[width=0.88\textwidth]{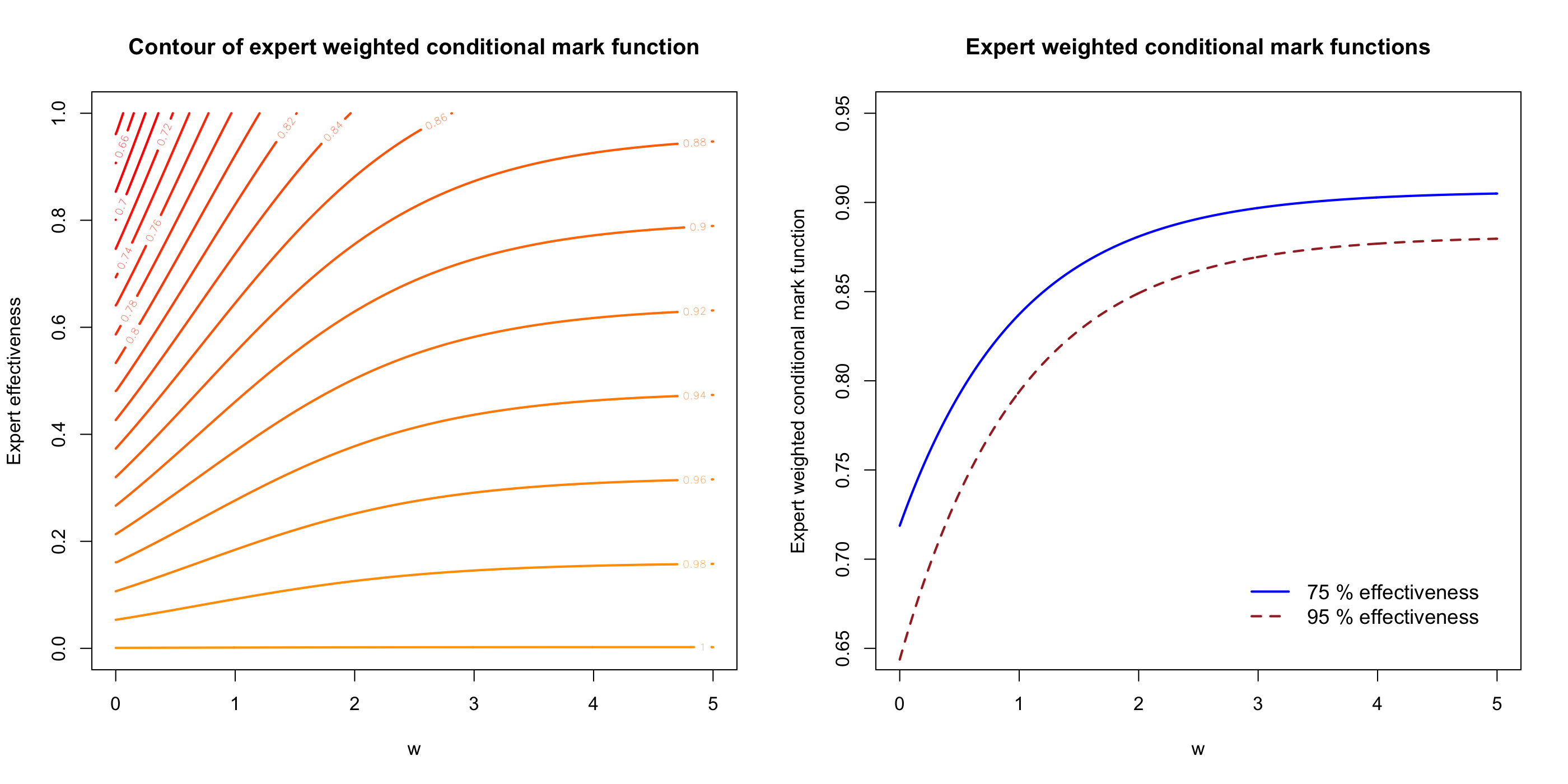}
\includegraphics[width=0.44\textwidth]{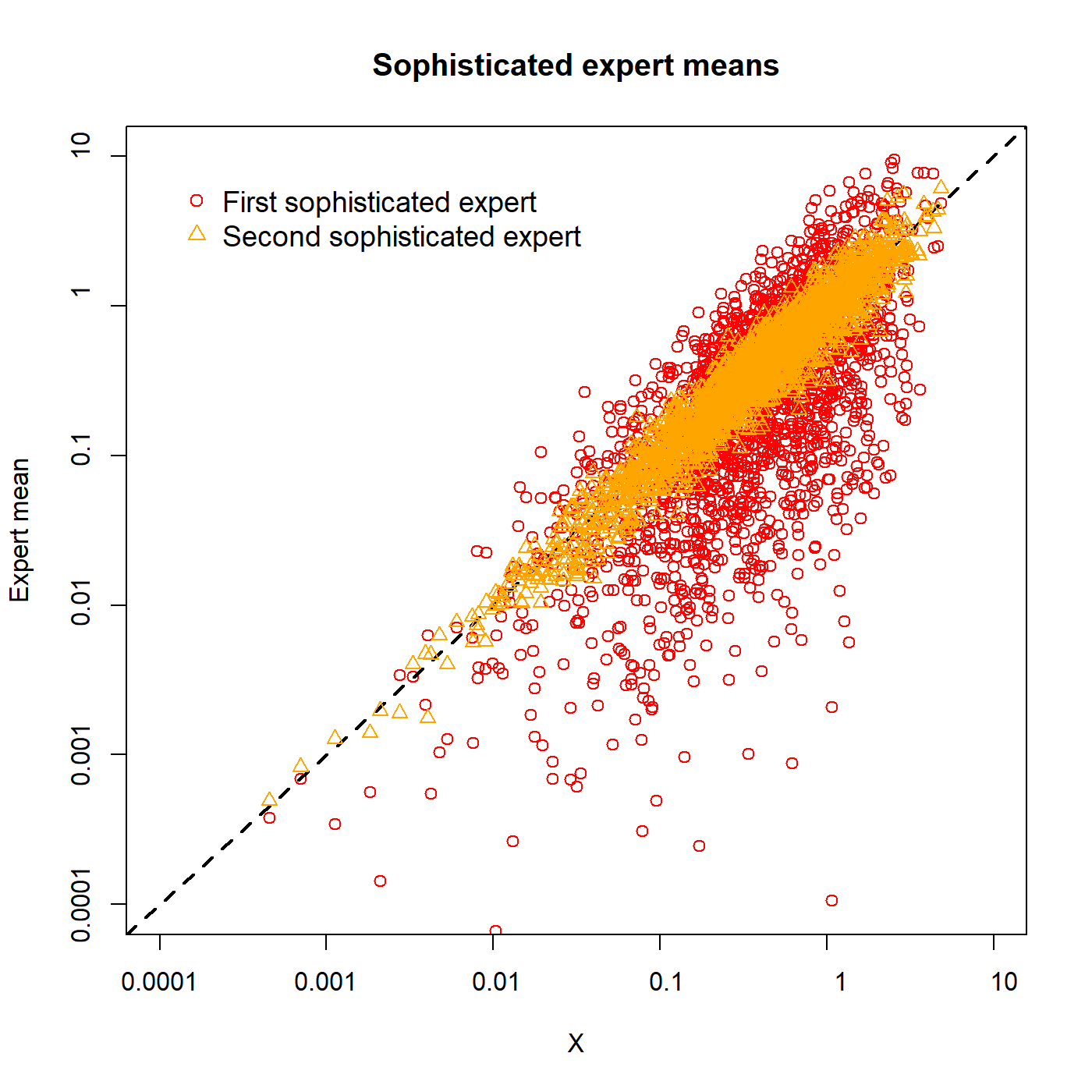}
\includegraphics[width=0.44\textwidth]{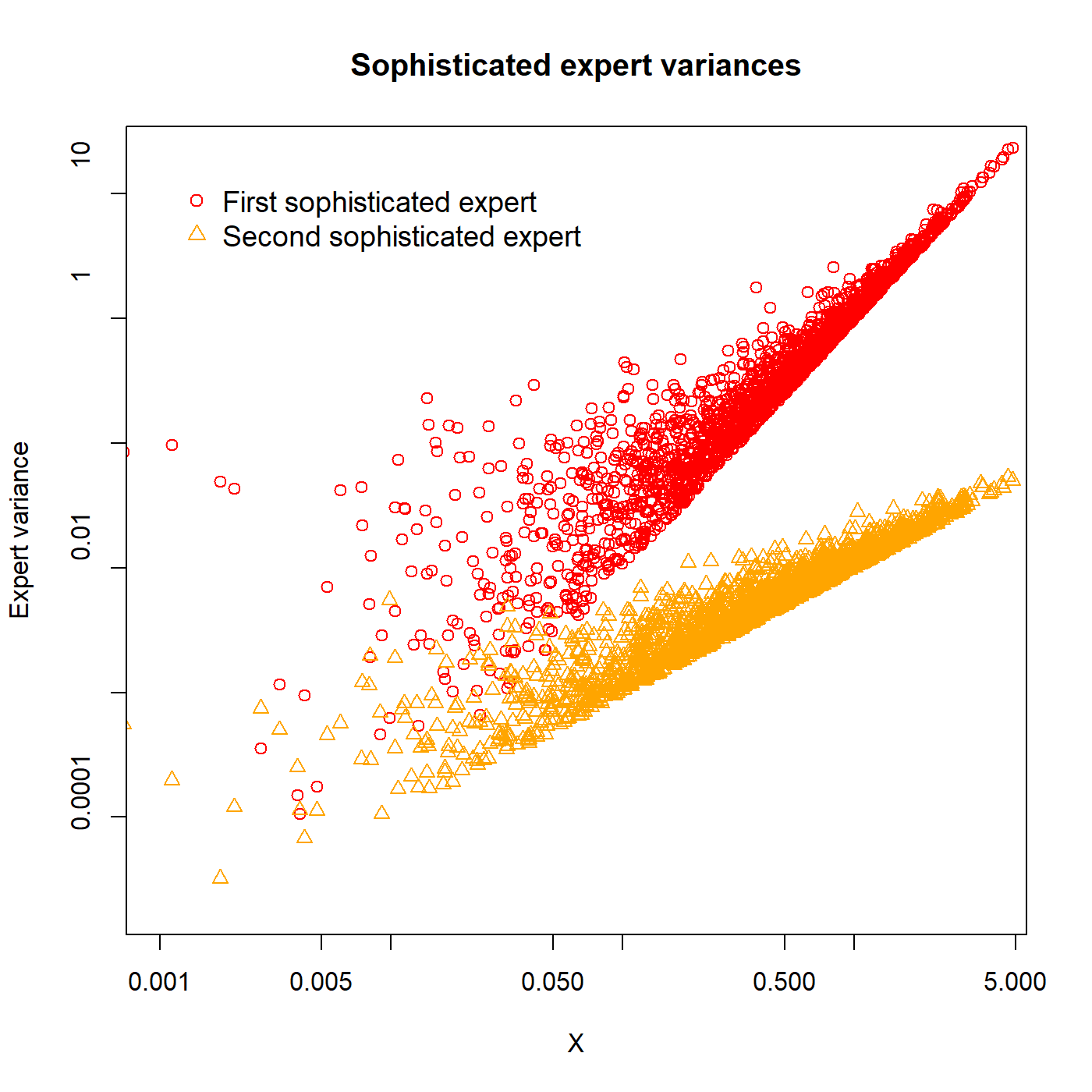}
\caption{Top left panel: contours of the expert weighted conditional mark function. Top right panel: expert weighted conditional mark function for the two experts considered in the simulation study. Bottom panels: mean and variance scatter plots for the kernels associated to the sophisticated experts.
}
\label{fig:new_sim_expert_info}
\end{figure}

The results of an analysis using our expert estimators are provided in Figure~\ref{new_sim_survival_curves}. The first, and anticipated, observation is that the usual product-limit estimator grossly underestimates the survival curve under contamination. Perhaps more subtly, but still in line with the theory, is that -- for our choice of parameters -- improved access to sophisticated expert information leads to much more sudden convergence to the true distribution, while improved access to a good estimate of $w \mapsto p(w)$ appears to have a more mellow and steady effect. Note that for very large durations, which are maybe of lesser importance for our application (but could be of interest in other domains){,} the sophisticated survival curves seem to wean off sharply as a consequence of the choice of kernel. Kernel choices (and their calibration) to capture appropriate tail behaviors are a subject of further research. {It should be stressed that we have only established strong consistency on compacts and, consequently, we cannot at this point guarantee convergence on the `maximal interval'.}

\begin{figure}[!htbp]
\centering
\includegraphics[width=0.8\textwidth]{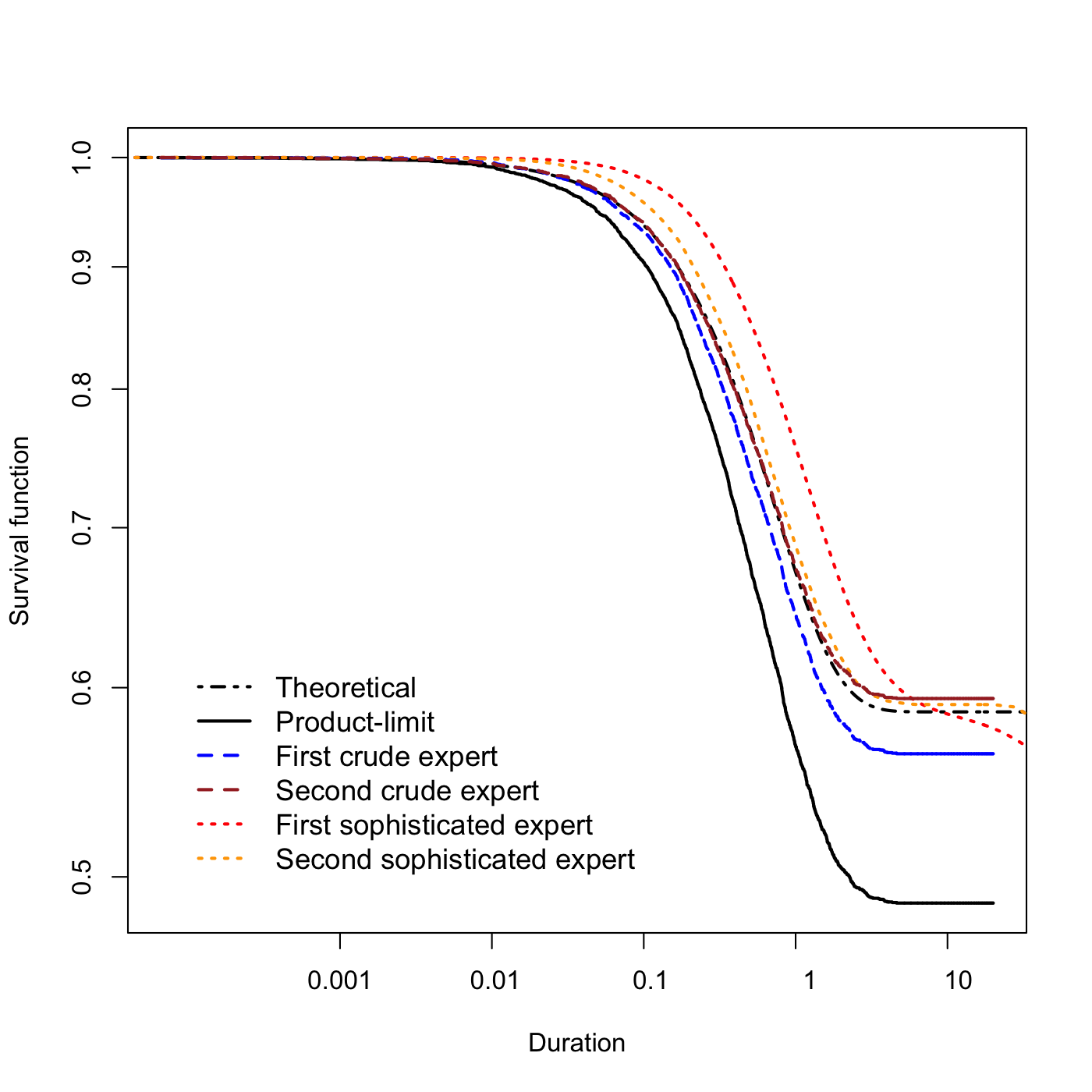}
\caption{Survival curves for the various expert estimators for the simulated disability duration dataset.
}
\label{new_sim_survival_curves}
\end{figure}

\pagebreak

\subsection{Real dataset} \label{sec:real}

We now consider the French temporary disability insurance dataset \texttt{freptftempdis} which may be found in the \texttt{R} package \texttt{CASdatasets}. We focus on the duration of disability for individuals with both temporary and permanent disability awards, with permanent disability counted as right-censoring at the latest observation time, which is ultimo $2018$. Furthermore, we look only at entries from $2015$ and only at entry ages between $55$ and $60$, so as to have a homogeneous population. Extensions including covariate information would allow modeling the entire dataset. Thus, we are left with a sample of $62{,}384$ individuals with roughly $13.7\,\%$ being right-censored data points. Figure~\ref{fig:real_data} shows a histogram of the durations, as well as the location of the right-censored observations within the entire sample; right-censoring is observed to be more common for larger durations.

\begin{figure}[!htbp]
\centering
\includegraphics[width=0.44\textwidth]{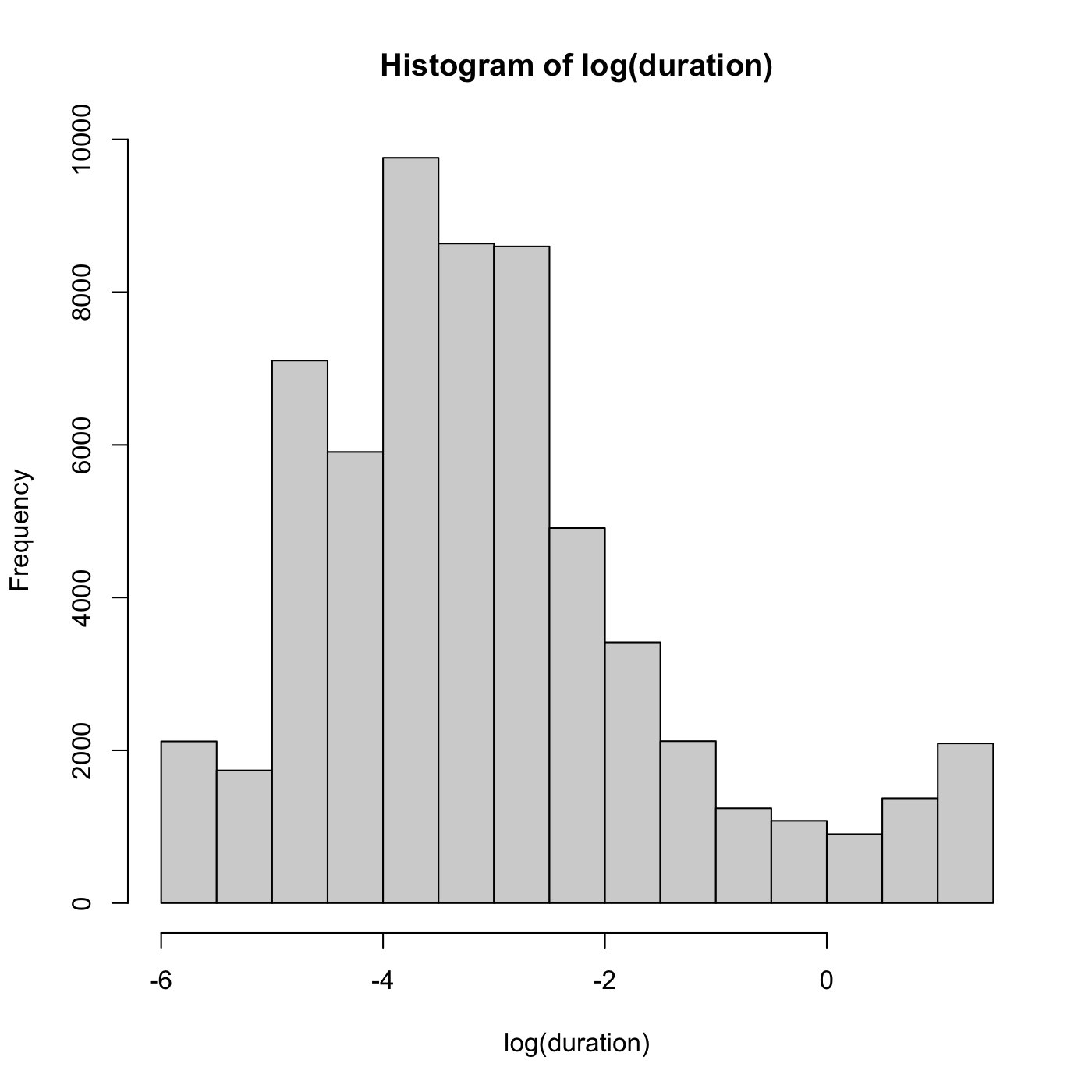}
\includegraphics[width=0.44\textwidth]{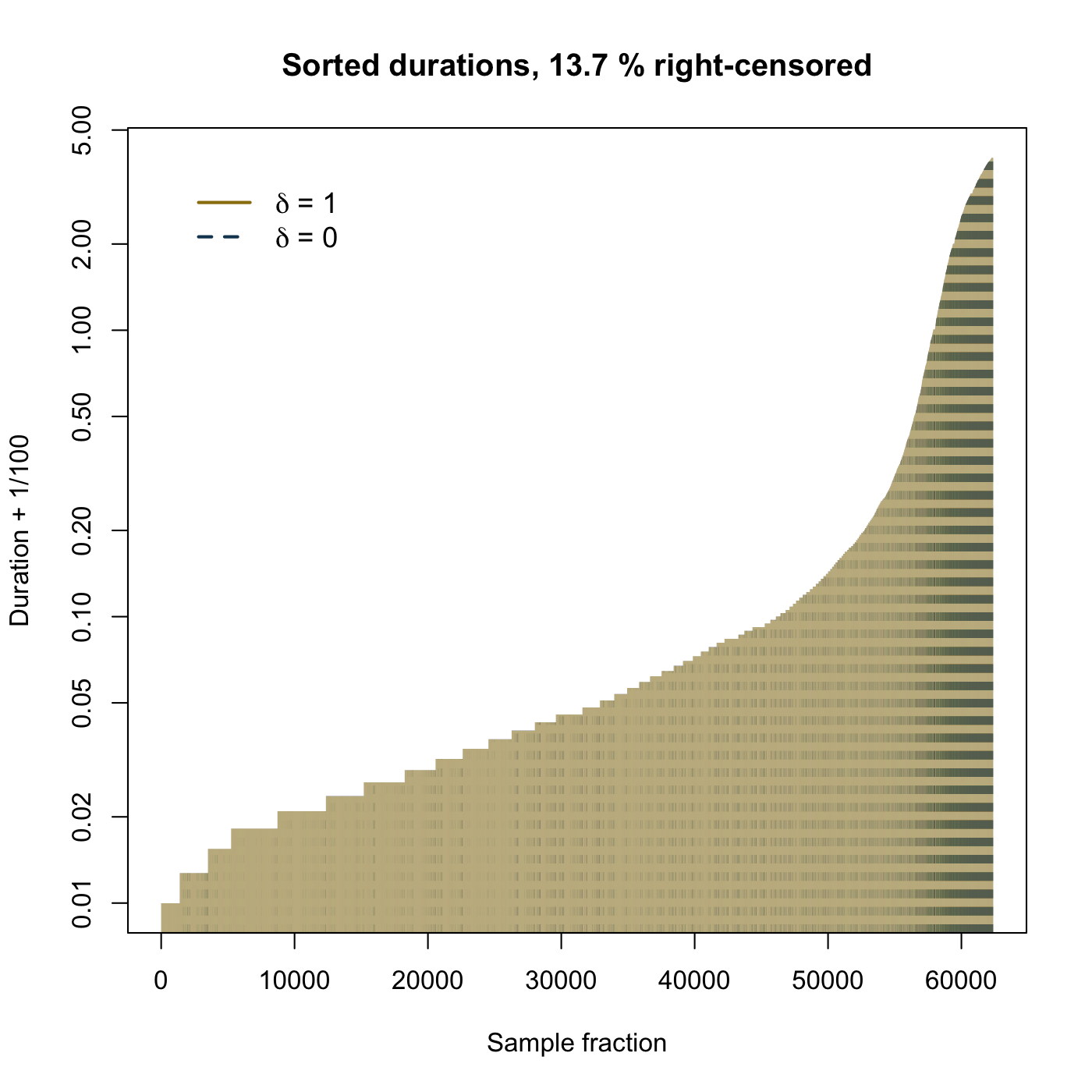}
\caption{Left panel: histogram of log-durations from the French temporary disability insurance dataset for entries in $2015$ and entry ages between $55$ and $60$. Right panel: sorted durations in log scale, with right-censored observations, corresponding to $\delta=0$, bundling up at larger sample fractions.
}
\label{fig:real_data}
\end{figure}

%Figure moved to this line of the LaTeX-code to improve visuals/layout.
\begin{figure}[!htbp]
\centering
\includegraphics[width=0.88\textwidth]{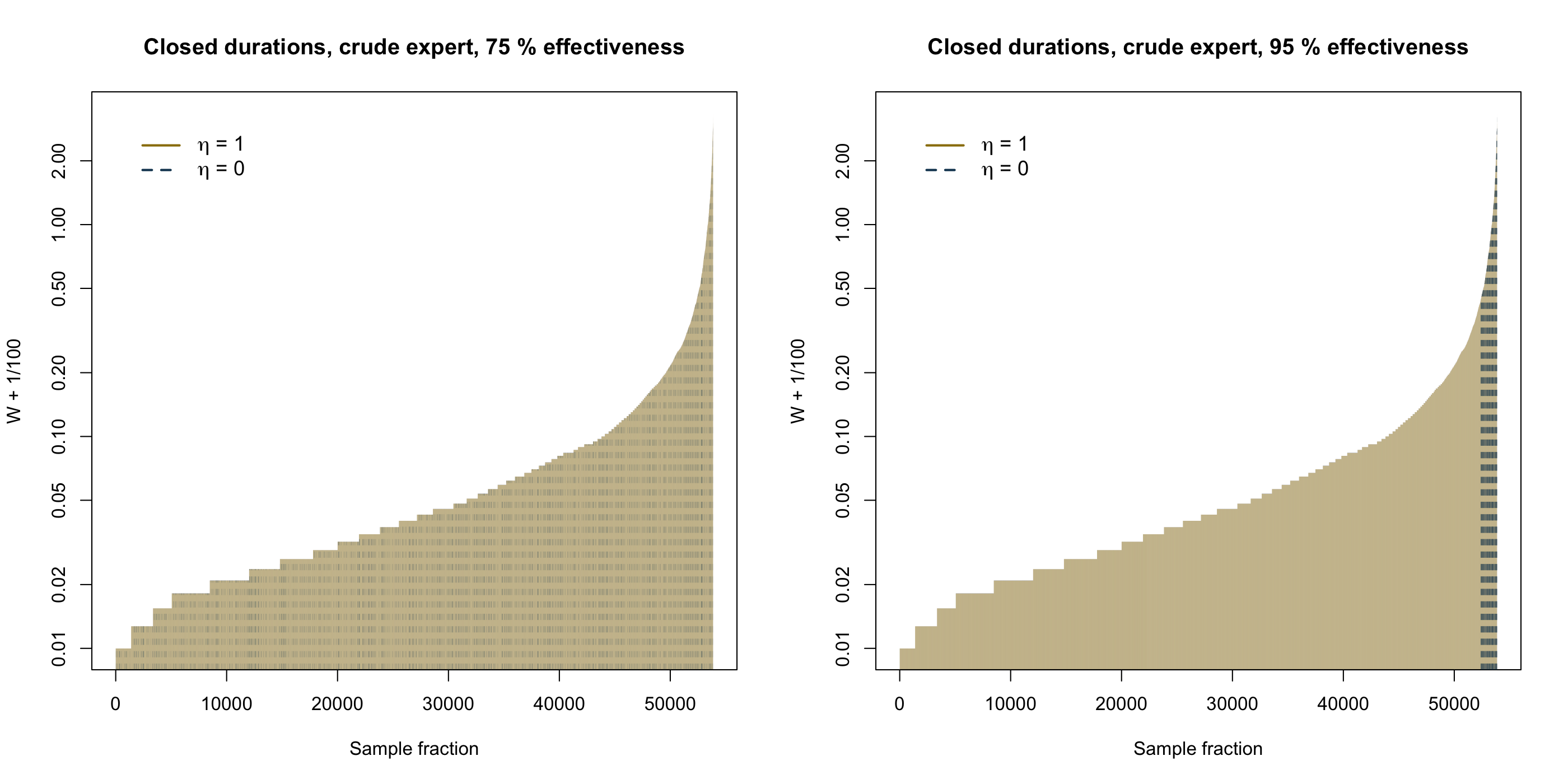}
\includegraphics[width=0.44\textwidth]{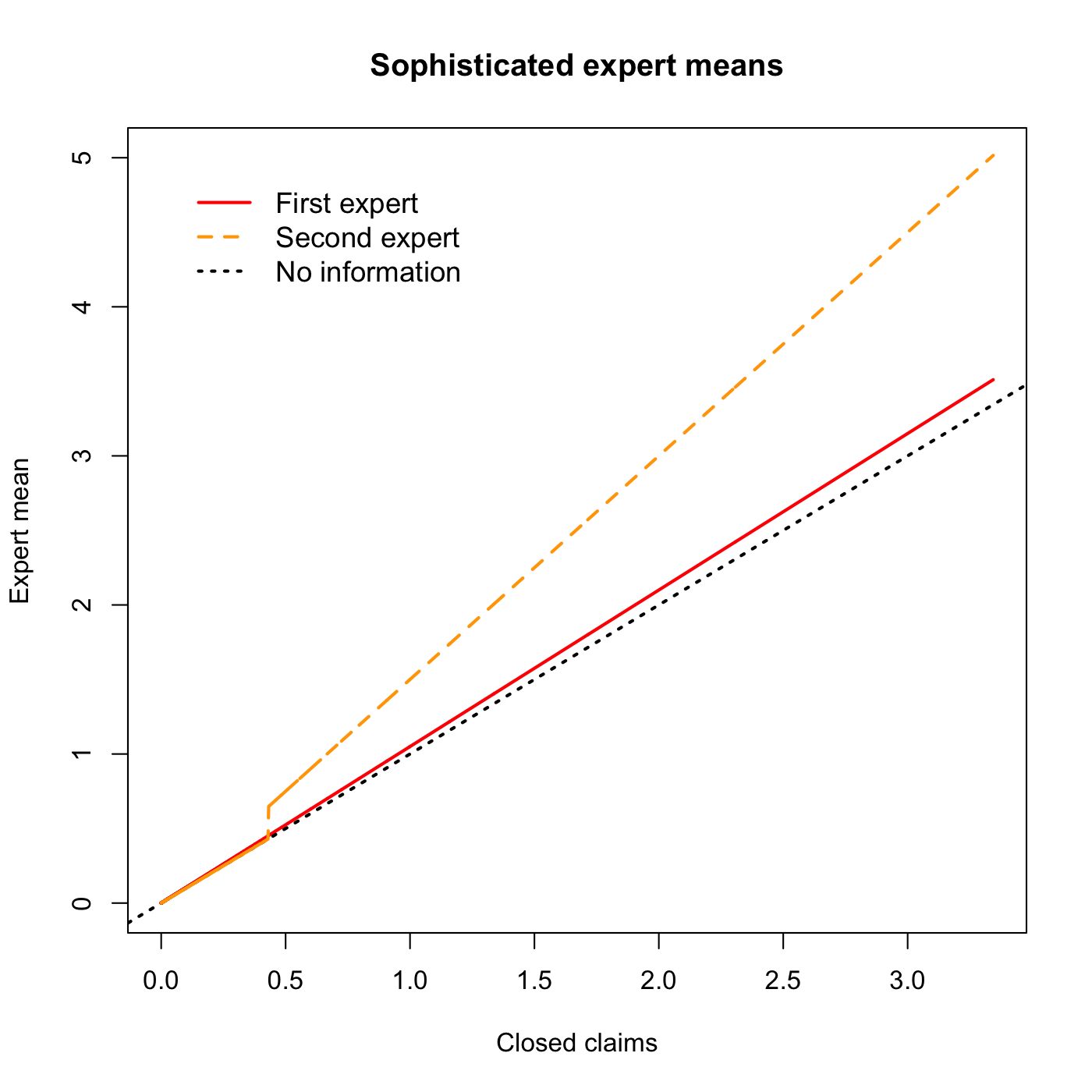}
\includegraphics[width=0.44\textwidth]{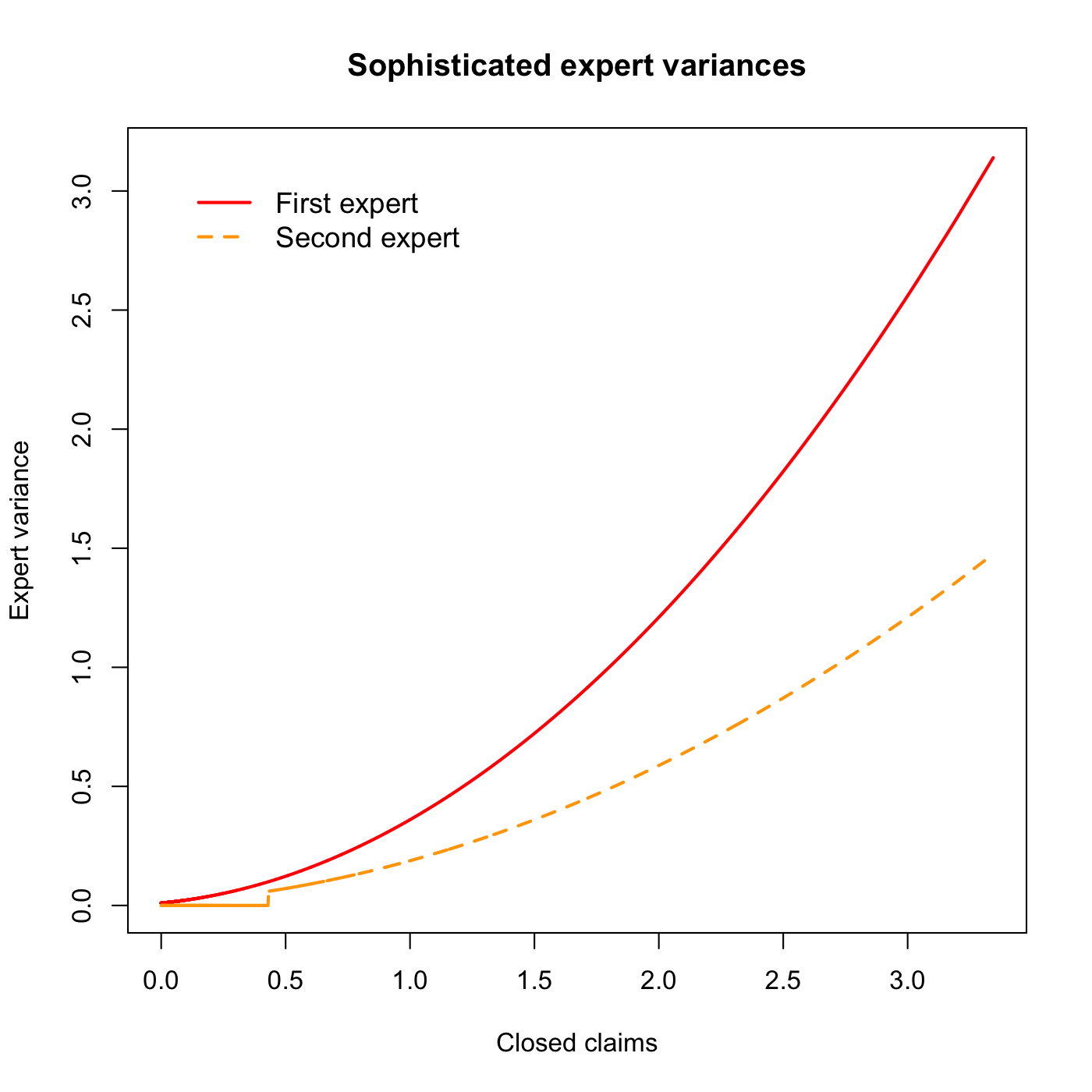}
\caption{Expert information for the French temporary disability insurance dataset. Top panels: crude expert information for closed durations. Bottom panels: sophisticated expert mean and variance specifications for closed claims.
} \label{fig:real_expert_info}
\end{figure}

\begin{figure}[!htbp]
\centering
\includegraphics[width=0.88\textwidth]{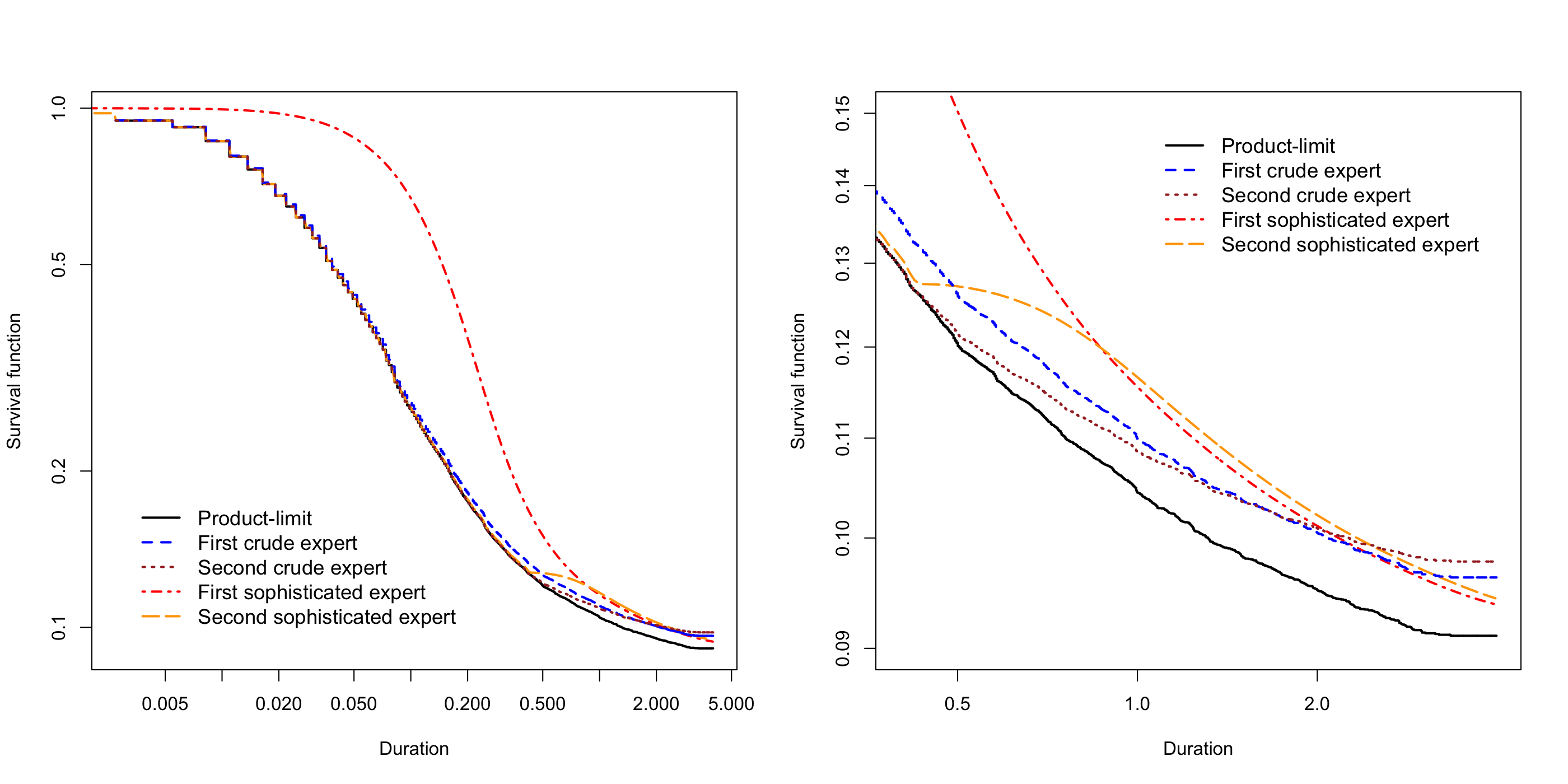}
\caption{Survival curves for the distinct expert information scenarios for the French temporary disability insurance dataset.
} \label{real_survival_curves}
\end{figure}

The next step in the analysis is to create realistic expert information scenarios to consequently see how the expert estimators behave. To this end, we distinguish between crude and sophisticated expert information. We describe below in detail the procedure we use to generate this information and summarize it in Figure~\ref{fig:real_expert_info}.

For the crude expert information, we consider two approaches. In the first one, we randomly select $2\,\%$ of the closed durations to reopen. That is we independently set $\eta_k=\delta_k \times B_k$, for $B_k\stackrel{iid}{\sim}\mbox{Bernoulli}(0.98)$. Regarding the second approach, we randomly select $20\,\%$ out of the longest $10\,\%$ of durations, that is, for all $k$ such that $W_k$ is above the empirical $90\,\%$ quantile we set $\eta_k=\delta_k \times B_k$, where $B_k\stackrel{iid}{\sim}\mbox{Bernoulli}(0.80)$. The latter approach reflects higher precision but reduced breadth of the expert. The two top panels of Figure~\ref{fig:real_expert_info} depict this construction.

Concerning the sophisticated expert information, we once again consider two cases. The first one considers {truncated Gaussian kernels with mean} $1.05\,W_k$ and with standard deviation $\frac{1}{10}+\frac{1}{2} W_k$, while the second only considers the top $10\,\%$ largest $W_k$ {and truncated Gaussian kernels with mean} $1.5\,W_k$ and with standard deviation $\frac{1}{10}+\frac{1}{3} W_k$. Here, as for the crude expert, we have again induced a precision versus breadth compromise. The two bottom panels of Figure~\ref{fig:real_expert_info} depict this construction.

From the four experts we have constructed above, we plot the associated survival estimators in Figure~\ref{real_survival_curves}. We see that although all estimators lift the tail of the distribution toward a heavier specification -- as desired -- they do so in different ways, some of which we briefly discuss. The second crude expert only provides information on the censoring indicators for large values of $W_k$, but the estimates for all quantiles is modified as a result. This feature occurs since the expert $\eta_k$ enter lower quantile estimates through the inverse probability of censoring weights. This is in contrast to the second sophisticated expert, which has a different tail only for upper quantiles and is otherwise behaving like the usual product-limit estimator for smaller quantiles. Finally, the first sophisticated expert has a particularly different (and smooth) curve from the rest of the estimators. This may be a concern for the cautious modeler; however, depending on the quality of the expert information, such a curve could actually be the least biased, and thus its unusual shape should not be an immediate deterrent to its use.

{
\section{Concluding remarks}\label{sec:con}

In this paper, we have shown that by integrating expert information into the classic Kaplan--Meier estimator, we can improve the estimation of survival curves in scenarios where some observations are not only right-censored but also subject to contamination. To be specific, we have incorporated expert judgments and beliefs in the form of binary variables or kernel specifications that lead to modified product-limit estimators. Under suitable conditions on the asymptotic quality of expert information, we have established strong consistency on compacts for these estimators. Our results have practical implications for actuaries and other practitioners working in the insurance and pensions industry, as well as in other fields where survival analysis is applied.

To keep the presentation to the point, we have not considered consistency on the `maximal interval' nor allowed for types of missingness other than right-censoring. In certain applications, observations are subject not only to right-censoring but also to left-truncation. We believe that independent left-truncation can be accommodated by properly adjusting the Turnbull estimator of~\cite{Turnbull1976} and by following the approach in~\cite{WangJewellTsai1986}. Regarding strong consistency for the tail of the distribution, a starting point could be~\cite{StuteWang1993}, where strong consistency on the `maximal interval' is established by utilizing a reverse supermartingale property of the estimator in $n$. We would like to think of this increase in mathematical sophistication as being, at least to a certain degree, unrelated to the inclusion of expert information. However, additional research would be required to confirm this claim, possibly only under additional assumptions.}

\section*{Acknowledgments}
Martin Bladt would like to acknowledge financial support from the Swiss National Science Foundation Project 200021\_191984{; most of his work was carried out while he was a postdoc at the University of Lausanne.} The collaboration between the authors was strengthened by research stays of Christian Furrer at the University of Lausanne facilitated by Hansj{\"o}rg Albrecher and partly supported by Dr.phil.\ Ragna Rask-Nielsens Grundforskningsfond. Christian Furrer has carried out this research in association with the project frame InterAct.

\bibliography{bladtfurrerexpert2022.bib}
\bibliographystyle{amsplain}

\end{document}